\documentclass[10pt]{article}
\usepackage{amsthm}
\usepackage{amsfonts}
\usepackage{amssymb}
\usepackage{amsthm}
\usepackage{amsmath}    
\usepackage{hyperref}
\usepackage{bbm}
\usepackage{graphicx}
\usepackage{latexsym}
\usepackage{mathrsfs}
\usepackage{bbm}
\usepackage{slashed}

\usepackage{geometry}
\DeclareMathAlphabet\EuFrak{U}{euf}{m}{n}	%  the Bold Euler Fraktur
\SetMathAlphabet\EuFrak{bold}{U}{euf}{b}{n}	%  gothic font

\newcommand{\cA}{\mathcal{A}}
\newcommand{\cB}{\mathcal{B}}
\newcommand{\cC}{\mathcal{C}}

\newcommand{\cE}{\mathcal{E}}
\newcommand{\cF}{\mathcal{F}}

\newcommand{\cH}{\mathcal{H}}

\newcommand{\cL}{\mathcal{L}}

\newcommand{\cP}{\mathcal{P}}

\newcommand{\cS}{\mathcal{S}}

\newcommand{\cU}{\mathcal{U}}

\newcommand{\rC}{\mathrm{C}}

\newcommand{\rAB}{{\mathrm{AB}}}
\newcommand{\rDHR}{{\mathrm{DHR}}}

\newcommand{\rc}{\mathrm{c}}

\newcommand{\bC}{\mathbb{C}}

\newcommand{\bI}{\mathbb{I}}

\newcommand{\bN}{\mathbb{N}}
\newcommand{\bQ}{\mathbb{Q}}
\newcommand{\bR}{\mathbb{R}}
\newcommand{\bU}{\mathbb{U}}
\newcommand{\bZ}{\mathbb{Z}}

\newcommand{\si}{\sigma}

\newcommand{\eps}{\varepsilon}

\newcommand{\io}{\iota}

\newcommand{\Lb}{{\left\langle \right.}}
\newcommand{\Rb}{{\left. \right\rangle}}

\newcommand{\bs}{\boldsymbol}

\newcommand{\supp}{\mathrm{supp}}% supporto di una funzione
  \theoremstyle{plain}
  \newtheorem{definition}{Definition}[section]
  \newtheorem{theorem}[definition]{Theorem}
  \newtheorem{proposition}[definition]{Proposition}
  \newtheorem{corollary}[definition]{Corollary}
  \newtheorem{lemma}[definition]{Lemma}
  
  \theoremstyle{definition}
  \newtheorem{remark}[definition]{Remark}

\begin{document}

\author{
\textsc{Claudio Dappiaggi$^{a}$ \ \ Giuseppe Ruzzi}$^{b}$  \ \    \textsc{Ezio Vasselli}$^{b}$\\[5pt]
\small{$^{a}$  Dipartimento di Fisica -- Universit{\`a} di Pavia }\\
\small{\& INFN, Sezione di Pavia -- Via Bassi 6, 27100 Pavia, Italy.}\\
\small{$^{b}$  Dipartimento di Matematica, Universit\`a di Roma ``Tor Vergata'',}\\
\small{Via della Ricerca Scientifica 1, I-00133 Roma,  Italy.}  \\
\small{\texttt{claudio.dappiaggi@unipv.it},} \\[5pt]
\small{\texttt{ruzzi@mat.uniroma2.it},} \\[5pt]
\small{\texttt{ezio.vasselli@gmail.com}}\\\\
}

\title{Aharonov-Bohm superselection sectors}
\maketitle

\begin{abstract}
We show that the Aharonov-Bohm effect finds a natural description in the setting of QFT on curved spacetimes 
in terms of superselection sectors of local observables. 

The extension of the analysis of superselection sectors from Minkowski spacetime to 
an arbitrary globally hyperbolic spacetime unveils the presence of a new quantum number labeling charged superselection sectors. In the present paper we show that this ``topological" quantum number 
amounts to the presence of a background flat potential which rules the behaviour of charges when transported along paths as in the Aharonov-Bohm effect. To confirm these abstract results we quantize the Dirac field in presence of a background flat potential and show that the Aharonov-Bohm phase gives an irreducible representation of the fundamental group of the spacetime labeling the charged sectors of the Dirac field.

%These abstract (model independent) results are confirmed by the analysis of Dirac field quantized in the presence of a background flat potential.

%
%In essence what we show is that switching on a background flat potential amounts to switch on the topological component of the sector. 
%
We also show that non-Abelian generalizations of this effect are possible only on spacetimes
with a non-Abelian fundamental group.
\end{abstract}

%\tableofcontents
\markboth{Contents}{Contents}

%\newpage

\section{Introduction.}
\label{intro}

The analysis of superselection sectors is one of the central results of algebraic quantum field theory.
It allows to derive from first principles particle statistics, the particle-antiparticle correspondence, quantum charges and, as a consequence, the appearance of global gauge symmetries, independently of the model under consideration
\cite{DR90}.

In the original formulation, superselection sectors were defined by Doplicher, Haag and Roberts in terms
of localized and transportable endomorphisms $\rho(o) : \cA \to \cA$ of the global observable $\rC^*$-algebra $\cA$, 
where $o$ is a double cone in Minkowski spacetime. In this way a new representation (sector)
\[
\pi : \cA \to \cB(\cH_0) \ \ \ , \ \ \ \pi := \pi_0 \circ \rho(o) \, ,
\]
is defined starting from a vacuum representation $\pi_0$.
The term \emph{localized} indicates that $\rho(o) \restriction \cA_e$ is the identity 
on any local algebra $\cA_e \subset \cA$ generated by observables localized in a double cone
$e$ causally disjoint from $o$, while \emph{transportable} means that, for any double cone $a$, $\rho(o)$  
is unitary equivalent to an endomorphism localized in $a$.  
This yields the DHR selection criterion \cite{Haa}.

\medskip

It was later recognized by Roberts that an equivalent formulation of superselection sectors could be achieved
by considering \emph{charge transporters}, that is, families $z$ of unitaries
$z(a,o) \in\cA_a$, $o \subset a$,
fulfilling a cocycle relation.
%
%\[
%z( \hat{o},a ) z(a,o) \ = \ z( \hat{o},o )
%\ \ \ , \ \ \ 
%o \subseteq a \subseteq \hat{o}
%\, ;
%\]
%
Each $z$ defines a family $\{ \rho^z(o) \}_o$ of localized endomorphisms, such that 
\[
z(a,o) \, \rho^z(o)(A) \ = \ \rho^z(a)(A) \, z(a,o) \ \ \ , \ \ \ A \in \cA \, .
\]
In Minkowski spacetime this approach turns out to be equivalent to the original one \cite{Rob0,Rob}.
Yet things change in the context of curved spacetimes $M$,
mainly due to topological obstructions that arise when the family of \emph{diamonds},
the analogues of double cones, is not upward directed under inclusion.
A first point is that the notion of an algebra of global observables is not well-defined
and must be replaced by the \emph{universal algebra}, indicated with $\vec{\cA}$ \cite{Fre},
having the property of lifting any family of representations
\[
\pi_o : \cA_o \to \cB(\cH) 
\ \ \ : \ \ \ 
\pi_a \restriction \cA_o \, = \, \pi_o
\, , \   o \subset a
\, .
\]
A second point is that localized endomorphisms alone cannot encode the whole physical content of sectors,
\cite{Ruz05,BR08}. Therefore the use of charge transporters becomes necessary in the more general scenarioof a curved spacetime $M$.

\medskip

A condition that must be imposed on $z$ in order to get a well-defined representation $\pi$ of $\vec{\cA}$ is 
\[
z_p \ := \ z(a,o_{n-1}) \cdots  z(o_2,o_1) z(o_1,a) \ \equiv \ \bI \, ,
\]
where $p := \{ a , o_1 , \ldots , o_{n-1} \}$ is a suitable covering (\emph{path-approximation}, \S \ref{Bb}) 
of an arbitrary loop (closed curve)  $\gamma : [0,1] \to M$. 
This condition is called \emph{topological triviality}, and the justification of this terminology
is given by the fact that $z_p$ only depends on the homotopy class of $\gamma$.
The question of what happens if topological triviality is not imposed then naturally arises.
An answer is given in a series of papers, in which the following facts were established:
\begin{enumerate}
\item Statistics, charge content and particle-antiparticle correspondence are well defined on generic 
      charge transporters. Thus they have all the properties that allow to
      assign a physical interpretation to topologically trivial transporters.
      Moreover, each $z$ defines a representation
      \begin{equation}
      \label{eq.intro1}
      \sigma : \pi_1(M) \to \bU(n) \, ,      
      \end{equation}
      which affects the transport of localized endomorphisms along loops.
      The integer $n \in \bN$ is called the \emph{topological dimension} \cite{BR08}. 
      The existence of such sectors for a massive Boson field in a 2-dimensional spacetime 
      has been proved in \cite{BFM09}.
\item Generic charge transporters do not define representations of $\vec{\cA}$. 
      Rather, they define representations of the net $\cA_{KM} := \{ \cA_o \}_o$ on flat Hilbert bundles over $M$,
      which correspond to covariant representations of a universal $C^*$-dynamical system
      $\alpha : \pi_1(M) \to {\bf aut}\cA_*$ 
      encoding parallel transport along loops \cite{RV11,RVCX,RVkhom}.
      It turns out that $\vec{\cA}$ is a quotient of $\cA_*$, corresponding to the set of representations of $\cA_*$
      with trivial action of $\pi_1(M)$ on the Hilbert space. 
\end{enumerate} 
The physical interpretation of these results was proposed by Brunetti and the second named author
in the scenario of the Aharonov-Bohm effect.
There, superposition of wavefunctions of charged particles is affected by phases (parallel transports) of type
\begin{equation}
\label{eq.intro2}
\sigma(\gamma) \ = \ \exp \oint_\gamma A \, ,
\end{equation}
where $A$ is a potential with vanishing electromagnetic field, and it is natural to note
the similarity with (\ref{eq.intro1}).
The suitable model in which such an interpretation can be tested at the level of quantum field theory 
is the Dirac field,
describing charged quantum particles in a curved spacetime $M$. 
A step in this direction has been made in \cite{VasQFT}, where it is
proven that introducing a background potential represented by a closed 1--form
yields a representation of the observable net $\cA_{KM}$ of the free Dirac field
over a flat Hilbert bundle over $M$.

\medskip

%[ev]
%
In the present paper we complete the above partial results and perform an analysis of 
the superselection structure of $\cA_{KM}$. Our main results show that:
\begin{itemize}
\item Superselection sectors $z$ define background flat potentials $A^z$, 
that are, at the mathematical level,
flat connections on suitable flat Hermitean vector bundles $\cL^z \to M$ whose rank coincides with the topological dimension of $z$ (Theorem \ref{Bd:9}).
We shall be interested, in particular, in those sectors with topological dimension 1, so that $\cL^z$ is a line bundle.
\item Superselection sectors with topological dimension 1 and "charge 1", in a sense that shall be clarified later, 
are in one-to-one correspondence with \emph{twisted field nets}. 
These are representations of the field net of the Dirac field on
flat Hilbert bundles over $M$ with monodromy given by the "second quantization"
of the monodromy of $\cL^z$ (Theorem \ref{thm.D2});
\item When no torsion appears in the first homology group of the spacetime, $\cL^z$ is topologically trivial.
In this case we are able to construct a twisted field solving the Dirac equation
with interaction term $A^z$, and the representation (\ref{eq.intro1}), with $n=1$, 
takes the form (\ref{eq.intro2}) (Theorem \ref{thm.D1}).
This occurs in physically relevant scenarios such as de Sitter and anti-de Sitter spacetimes,
as well as the spacetime complementary to the ideally infinite solenoid of the Aharanov-Bohm apparatus.
\end{itemize}
Both the twisted field net and the twisted Dirac field are indistinguishable from 
the analogous untwisted objects in simply connected regions of $M$.
This is the reason way $\cA_{KM}$ is able to detect them.
The non-trivial parallel transport carried by the family of localized endomorphisms defined
by a sector is interpreted in terms of the background flat potential, shedding some light on the role that the corresponding interaction plays on quantum charges.

\medskip

\noindent The paper is organized as follows.

In \S \ref{B} we recall results from \cite{BR08,Ruz05} describing the superselection structure,
say $Z^1(\cA_{KM})$, of a generic observable net. Moreover, by applying a result by Barrett \cite{Barr},
for any sector $z$ we exhibit a flat $\mathfrak{u}(n)$-connection such that (\ref{eq.intro1}) is 
the associated parallel transport. 
%defined by the path-ordered exponential $\exp \oint A^z$.
This yields a Chern character defined on $Z^1(\cA_{KM})$ with values in the odd cohomology 
$H^{odd}(M,\bR/\bQ)$ of the spacetime (Remark \ref{rem.Chern}).
The first component of the Chern character will be shown to correspond to Aharonov-Bohm phases
in the subsequent sections.

In \S \ref{C} we focus on the case in which $\cA_{KM}$ is the observable net of the free Dirac field. 
By applying results in \cite{BJL02}, we show that Haag duality holds in the GNS representation
$\pi_\omega$ of any pure quasi-free state $\omega$ of the Dirac-CAR algebra defined 
on the spinor space of a globally hyperbolic spacetime (Theorem \ref{ThtHaag}).
This, and the $III$-property fulfilled by the local Von Neumann algebras living in $\pi_\omega$, 
\cite{dAH06}, yields two key properties for the analysis of sectors of $\cA_{KM}$. 
Finally, we construct a family of topologically trivial sectors of $\cA_{KM}$
labeling the electron-positron charge (Remark \ref{rem.EC}).

In \S \ref{sec.DA} we construct Dirac fields interacting with background potentials $A$, which are closed 1--forms ($dA=0$).
These fields are ``twisted", as in \cite{Ish}, namely they are defined on the space of sections of 
the Dirac bundle tensor a flat line bundle with connection $A$.
We show that these fields are in one-to-one correspondence with sectors in $Z^1(\cA_{KM})$ such that $\cL^z$ is 
topologically trivial (Theorem \ref{thm.D1}).
The parallel transport defined by $A$ appears, from one side as a byproduct of the interaction term in the Dirac equation,
and on the other side as the parallel transport of DHR-charges (localized endomorphisms) of $\cA_{KM}$. 

In \S \ref{sec.E} we draw our conclusions and we illustrate further developments, in particular for
$\pi_1(M)$ non-Abelian and generalized Dirac fields with non-Abelian gauge group.

In Appendix \S \ref{App} we prove a technical result, allowing to adopt the usual Haag duality instead of the
punctured one in the analysis of sectors.

Some of the results of the present paper have been presented in \cite{VasBPS}
in a simplified form, in particular by avoiding to discuss Haag duality
and without going in details on the structure of twisted Dirac fields.
%
% EV:
% Scrivere nel corpo dell'articolo che \pi_o := \rho^z(o)_o, dove \pi_o è definito in VasBPS 

\section{Aharonov-Bohm effects in terms of sectors}
\label{B}

Aim of the present section is to describe the Aharanov-Bohm effect in terms of superselection sectors
of the net of local observables on a 4-d globally hyperbolic spacetime $M$. From the intrinsic viewpoint of the electromagnetic field and of the vector potential, this effect has been studied instead in \cite{Ben14,San14}. \smallskip

We start by recalling some basic facts on globally hyperbolic spacetimes and we introduce the net of local observables defined in a reference representation. 
We then discuss 
DHR-charges of the observable net in terms of charge transporters. These are 
cocycles depending on a suitable family of regions of the spacetime,  diamonds, taking values in the unitary group of the observable net. 
The key observation is that diamonds encode the fundamental group $\pi_1(M)$ of the spacetime $M$ and cocycles define representations of $\pi_1(M)$. The Aharonov-Bohm effect is manifested in cocycles which 
define nontrivial representations of the fundamental group of the spacetime. We shall prove that 
these cocycles are nothing but the holonomy of smooth flat connections 
acting on a DHR charge. The form of the action is ruled by the charge quantum numbers. \smallskip
  
The material appearing in the present section is essentially a convenient summary of \cite{BR08,Ruz05}.
Exceptions are Lemma \ref{Bd:6} and Theorem \ref{Bd:9}, where the interpretation in terms of flat connections is given.

%
%
%
%
% 
%%
%As a second step we consider the net  of local observables, and show how a class of charged representations of the net can be selected 
%from a reference representation satisfying the microlocal spectrum condition.
%Then we introduce the charge transporters associated with these representations which,
%from the mathematical point of view, are 1-cocycles associated to the family of diamonds. 
%However we can distinguish two types of 1-cocycles: 
%%
%\begin{itemize}
%\item Topologically trivial 1-cocycles (\emph{DHR-cocycles}), defining trivial representations of $\pi_1(M)$. 
%      These 1-cocycles describe charged sectors (\emph{DHR-charges}) and correspond to those of the DHR analysis \cite{GLRV01}.
%      %
%\item And 1-cocycles which are not topologically trivial: these decompose as a composition of a DHR-cocycle
%      and a 1-cocycle describing a non-trivial representation of $\pi_1(M)$. 
%      We shall prove that this second cocyle is nothing but the path ordered integral of a smooth flat connection 
%      acting upon a DHR charge. The form of the action is ruled by the charge quantum numbers.
% \end{itemize}
%%
%The material appearing in the present section is essentially a convenient summary of \cite{Ruz05,BR08}.
%The exceptions are Lemma \ref{Bd:6} and Theorem \ref{Bd:9}, where the interpretation in terms of flat connections is given.
%

\subsection{Local observables in a globally hyperbolic spacetime}
\label{Ba}

Let $M$ denote a 4--dimensional, connected, globally hyperbolic spacetime. 
\emph{Causal disjointedness} relation  is a symmetric 
binary relation $\perp$ defined on subsets of $M$ as follows: 
\begin{equation}
\label{Ba:0}
o\perp\tilde o \ \ \iff \ \ o\subseteq M\setminus J(\tilde o)\ ,
\end{equation}
where $J$ denotes the causal set of $o$. The causal complement 
of a set $o$ is the open set $o^\perp:=M\setminus cl(J(o))$. An
open set $o$ is causally complete whenever $o=o^{\perp\perp}$. \smallskip

A distinguished class of subsets of $M$ is given by the set $KM$ of  \emph{diamonds} $o \subset M$ \cite[\S 3.1]{BR08},
which satisfy useful requirements:
they are open, relatively compact, connected and simply connected, causally complete,
and the causal complement $o^\perp$ is connected. 
Moreover, diamonds form a base of neighbourhoods for the topology of $M$.\smallskip 

In the present paper we are interested in describing the observable net in a \emph{reference} 
representation playing the same r\^ole as the vacuum representation in Minkowski spacetime.
This is defined in terms of a net $\cA_{KM}$ given by the assignment of a von Neumann algebra 
$\cA_o \subset \cB(\cH_0)$ for any $a \in KM$. We assume the following standard properties:
\begin{itemize}
\item (\emph{Isotony})  $\cA_o\subseteq \cA_a$ for $o\subseteq a$,
\item (\emph{Causality}) $\cA_o \subseteq \cA_a^\prime$ for $o\perp a$,
      where $\cA_a^\prime$ stands for the commutant of $\cA_a$,
\item (\emph{The Borchers property}) for any inclusion $cl(o)\subset \tilde o$,
      if $E$ is a non vanishing orthogonal projection of $\cA_o$ then there is an isometry 
      $V\in\cA_{\tilde o}$ such that $VV^*= E$;
\item (\emph{Haag duality}) $\cA_o = \cap \big\{\cA'_a\, :\, a \perp o\}$, for any $o\in KM$.
\end{itemize}
In addition $\cA_{KM}$ is assumed to be
\emph{irreducible} $\left(\cup_o\cA_o\right)^\prime=\mathbb{C}\mathbbm{1}$, 
and \emph{outer regular}  $\cA_o= \cap_{cl(o)\subset a} \cA_{a}$.
%
%\item  (\emph{punctured Haag duality}) \ \ \ \   
%given a point $x\in M$ there holds 
%%
%\[
% \cA_o
%= \cap \big\{\cA^\prime_{a} \, : \,  a\in KM, \ a \perp o, \ \ 
%         cl(a)\perp \{x\}\big\},
%\]
%%
%for any $o\in KM$ with $cl(o)\perp \{x\}$.

\medskip

Meaningful examples satisfying the above properties
are the the observable net of the free scalar field \cite{Ver97} and, as we shall see, and that of the free Dirac field, in representations induced by pure quasi-free Hadamard states. 
Anyway the above properties are expected to hold for any Wightman field over $M$.

\subsection{1-cocycles of the observable net and topology} 
\label{Bb}

We introduce the mathematical structure underlying the charges studied in \cite{BR08}.
The main fact is that the theory is encoded by charge transporters which define representations of the observable net by means of Haag duality.
A crucial mathematical property is that charge transporters satisfy a cocycle equation with respect to the partially ordered set (poset)  $KM$ and, as a consequence, give a representation of the fundamental group of the spacetime. In the present paper we give an equivalent but simplified
exposition of cocycles which do not relies on simplicial sets.\medskip

We start by recalling some elements of the "geometry" of $KM$, the poset of diamonds ordered under inclusion $\subseteq$.
A pair $a,o \in KM$ is said to be \emph{comparable} whenever either $o\subseteq a$ or $a\subseteq o$,
and in this case we write $a \gtrless o$.
Any comparable pair $o \gtrless a$ defines two oriented "elementary" paths: $(a,o):o\to a$, starting from $o$ and ending to $a$, and its \emph{opposite} $\overline{(a,o)}:=(o,a):a\to o$. We call the \emph{support} of $(a,o)$ the set $|(a,o)|:=a\cup o$. 
Arbitrary paths are obtained by composing elementary paths: for instance, given $(a,o)$ and $(o,\tilde a)$ we can form the path 
\[
p:=(a,o)*(o,\tilde a):\tilde a \to a
\]
and the opposite $\overline{p}:a\to\tilde a $ of $p$ as the path 
\[
\overline{p}:=\overline{(a,o)*(o,\tilde a)}:=
\overline{(o,\tilde a)}*\overline{(a,o)}= (\tilde a, o)* (o,a)  \, .
\]
The composition and the operation of taking the opposite extend 
to arbitrary paths in an obvious way:  
given $p:a\to o $ and $q:o\to\tilde o$ then 
\[
q*p:a\to \tilde o \  \ ,  \ \ \overline{q*p}= \overline{p}*\overline{q}: \tilde o \to a \ . 
\]
A path of the type $p:o\to o$ is called a \emph{loop} over $o$. 
It is verified that $M$ is connected if, and only if, $KM$ is \emph{pathwise connected} 
meaning that any pair of diamonds $a,o$ can be joined by a path $p:a\to o$ \cite{Ruz05}.
There is a homotopy equivalence relation $\sim$ on the set of paths:
the quotient by $\sim$ of the set of loops over a fixed $a\in KM$ gives the homotopy group $\pi_1(KM,a)$ which does not depend, up to isomorphism, 
on $a$, because of pathwise connectedness.
The isomorphism class is called the \emph{fundamental group} of $KM$, written $\pi_1(KM)$.
A key result is that, given $o\in KM$ and $x_o\in o$, there is an isomorphism
\begin{equation}
\label{Ba:1}
 \pi_1(M,x_o) \ni [\gamma] \mapsto [p_\gamma]\in\pi_1(K,o) \ ,
\end{equation}
where $ \pi_1(M,x_o)$ is the homotopy group of $M$ based at $x_o \in o$ \cite{Ruz05}. 
In the previous expression $p_\gamma$ is a \emph{path-approximation} of the closed curve 
$\gamma:[0,1]\to M$: that is, a path 
\[
p_\gamma=(o_n,o_{n-1})*\ldots *(o_2,o_1)*(o_1,o_0)
\]
such that there is a partition 
\[
t_0=0<t_1\dots < t_{n-1}<t_n=1 \ \ : \ \ 
\gamma(t_i)\in o_{i} \ {\mathrm{and}} \ \gamma([t_{i-1},t_i])\subset o_{i-1}\wedge o_i \, , \, \forall i \, .
\]
Thus closed curves $\gamma$ and $\beta$ are homotopic, i.e. $[\gamma] = [\beta]$, if and only if $p_\gamma\sim p_\beta$, i.e. $[p_\gamma] = [p_\beta]$.

\medskip

\begin{remark}
We highlight a few consequences of the previous result which are of relevance to this paper. 
If $KM$ is directed under inclusion, then $\pi_1(KM)$ is trivial \cite{Ruz05}.
So, if $\pi_1(M)$ is not trivial then $KM$ is not upward directed. 
This is actually the physical situation we are interested in,
since Aharonov-Bohm type effects appear when the spacetime has a non-trivial fundamental group.
\end{remark}

\medskip

We now introduce the set $Z^1(\cA_{KM})$ of 1-cocycles of the observable net, i.e.\ the charge transporters. 
A \emph{1-cocycle of the net $\cA_{KM}$} is a map assigning to any \emph{comparable} pair $o \gtrless a \in KM$
a unitary  operator $z(a,o)\in \cA_{|a,o|}$,
fulfilling the property $z(a,o)=z(o,a)^*$ and the \emph{cocycle equation}
\begin{equation}
\label{Bb:1}
 z(\hat{o},o)= z(\hat o,a)\, z(a,o) \ , \qquad o\subseteq a\subseteq \hat o \, .
\end{equation}
The previous definition implies $z(o,o)=1$ for any $o \in KM$, since $z(o,o)=z(o,o)z(o,o)$ and $z(o,o)$ is unitary. 
1-cocycles are extended to paths by defining
\[
z(p):=z(o,o_{n})\cdots z(o_2,o_1)z(o_1,a) \ , \qquad p=(o,o_{n})*\cdots* (o_{2},o_1)*(o_1,a) \ .
\]    
Note that $z(\overline{p})=z(p)^*$.
Any $z\in Z^1(\cA_{KM})$ preserves the homotopy equivalence relation, so that the mapping
$[p] \mapsto z(p)$, $[p]\in\pi_1(KM,o)$,
is well-defined. By (\ref{Ba:1}), we get the representation
\begin{equation}
\label{Bc:1}
\si_z : \pi_1(M,x_o) \to \cU(\cH_0)
\ \ \ , \ \ \ 
\si_z([\gamma]) := z(p_\gamma) \ , \qquad [\gamma]\in\pi_1(M,x_o)
\, ,
\end{equation}
where $p_\gamma:o\to o$ is a path-approximation of $\gamma$.
A 1-cocycle $z$ is said to be \emph{topologically trivial} whenever $\si_z$ is the trivial representation.

\smallskip

An \emph{intertwiner} between two 1-cocycles $z,\tilde z$ is a map
$t_a \in \cA_a$, $a \in KM$,
fulfilling the relations
$t_{a}\, z(p)=\tilde z(p)\, t_{o}$ for all $p:o\to a$.
The set of \emph{intertwiners} between $z,\tilde z$ is denoted by $(z,\tilde z)$.
The 1-cocycles $z,\tilde z$ are said to be \emph{unitarily equivalent} if there exists $t\in(z,\tilde z)$
such that $t_a$ is unitary for any $a \in KM$.
The 1-cocycle $z$ is said to be \emph{irreducible} whenever $(z,z)=\mathbb{C} \bI$, where $\bI \in \cB(\cH)$
is the identity,
and \emph{trivial} if it is unitarily equivalent to the trivial 1-cocycle $\io(a,o) \equiv \bI$.
It is easily seen that cocycles form the set of objects for a category $Z^1(\cA_{KM})$,
with the corresponding sets of intertwiners as arrows.
We denote the full subcategory of $Z^1(\cA_{KM})$ whose objects are topologically trivial 1-cocycles 
by $Z^1_t(\cA_{KM})$ \smallskip

\subsection{Charge quantum numbers} 
\label{Bc}
The physical content of 1-cocycles is encoded in the charge quantum numbers associated to their equivalence classes. 
In this section we shall recall how these quantum numbers arise and explain the nature of these charges. As already mentioned, 
these results have been proved in \cite{Ruz05,BR08}. There is, however, an important difference: in the present 
paper we assume Haag duality, whilst in the above references a stronger form of duality is used, the \emph{punctured} Haag duality.
In the appendix we show succinctly that Haag duality suffices.\smallskip

Using the defining properties of the observable net one can introduce a tensor product $\times$ (charge composition)
and a permutation symmetry $\eps$ (statistics), making $Z^1(\cA)$ and $Z^1_t(\cA)$ symmetric, tensor $\rC^*$-categories. In particular, one can identify a subset of 1-cocycles called \emph{objects with finite statistics} which turn out to be a finite direct sum of irreducible 1-cocycles.
The unitary equivalence class of any such irreducible object $z$ is classified by the following charge quantum numbers: 
\begin{itemize}
\item the \emph{statistical phase}  $\kappa(z)\in\{-1,1\}$ distinguishing between Fermi and Bose statistics; 
\item the \emph{statistical dimension} $d(z)\in \mathbb{N}$ giving the order of the (para)statistics; 
\item the \emph{charge conjugation} assigning to $z$ a conjugated irreducible 1-cocycle $\bar{z}$ 
      which has the same statistical phase and dimension as $z$.
\end{itemize}
We stress that objects with statistical dimension $1$ obey the  Bose/Fermi statistics according to their statistical phase.
We denote the subcategory of \emph{1-cocycles with finite statistics} by  $Z^1_\rAB(\cA)$ 
and its subcategory of \emph{topologically trivial 1-cocycles} by $Z^1_\rDHR(\cA)$. \smallskip  

To complete the physical interpretation  we need to observe that any 1-cocycle $z\in Z^1_\rAB(\cA)$ splits in two parts: the charged and the topological component. 
To see this we need a \emph{path frame} $P_e$ consisting of a collection of arbitrary paths 
$p_{ae}:e\to a$
joining a fixed diamond $e$, the pole,  with any $a$, and such that $p_{ee}$ is the trivial path $(e,e)$.
Once a path frame $P_e$ is given, \emph{the charged component} of $z$ is 
a topologically trivial 1-cocycle $z_{\mathrm{c}}$ in $Z^1_\rDHR(\cA_{KM})$ defined as 
\begin{equation}
\label{Bd:1}
z_{\mathrm{c}}(\tilde a,a):= z(p_{\tilde ae}*p_{ea}) \ , \qquad     a \gtrless \tilde a  \, ,
\end{equation}
where $p_{e\tilde a}$ is the opposite $\overline{p_{\tilde a e}}$ of $p_{\tilde a e}$. This cocycle encodes 
the charge content of $z$ as
the maps $z\mapsto z_{\mathrm{c}}$ and  $(z,\tilde z)\ni s \to s \in(z_{\mathrm{c}}, {\tilde z_{\mathrm{c}}})$   
define a functor 
\begin{equation}
\label{eq.charge}
Z^1_\rAB(\cA_{KM}) \to Z^1_\rDHR(\cA_{KM}) \, ,
\end{equation}
preserving the tensor product, the permutation symmetry and, hence, the conjugation\footnote{The functor is not full, as we may have $(z,\tilde z)\subsetneq (z_{\mathrm{c}}, {\tilde z_{\mathrm{c}}})$ \cite{BR08}}. On the other end, setting 
\begin{equation}
\label{Bd:2}
u_z(\tilde a,a) := z(p_{e\tilde a}*(\tilde a,a)*p_{ae}) \ , \qquad  \tilde a \gtrless a  \ , 
\end{equation}
one gets a 1-cocycle $u_z$, \emph{the topological component} of $z$,  taking values in the algebra of the chosen  pole $e$. i.e. \  
$u_z(\tilde a,a) \in \cA_e$ for all $\tilde a \gtrless a$. It encodes the topological content of $z$ since it defines the same representation (\ref{Bc:1}). Finally, the elements of $Z^1_\rAB(\cA_{KM})$ 
are completely characterized by their topological and the charged component since  
\begin{equation}
\label{Bd:3}
z(\tilde a,a) = (u_z\Join z_{\mathrm{c}})(\tilde a,a) =  
\alpha_{\tilde ae}(u_z(\tilde a,a))\cdot z_{\mathrm{c}}(\tilde a,a) \ ,\qquad 
\tilde a \gtrless a  \ , 
\end{equation}
where
$\alpha_{\tilde ae}(A):= z_{\mathrm{c}}(p_{\tilde ae})\, A \,z_{\mathrm{c}}(p_{\tilde ae})^*$. The composition 
$\Join$ is called the \emph{join}.   
These constructions do depend neither on the 
choice of the pole $e$ nor on that of the path frame $P_e$: different choices lead
to equivalent, in the corresponding categories, charged and topological components.\smallskip 
%%
%\begin{remark}
%Two observation are in order. The above constructions do not depend on the 
%choice of the pole $o$ and of the path frame $P_o$. A different choice leads
%to equivalent, in the corresponding categories, charged and topological components.
%The operation of join has a more general application than that given in 
%this paper. In fact topological trivial 1-cocycles and the 1-cocycle
%associated with a representation of the fundamental group, under suitable conditions, 
%can be joined to form a 1-cocycle $z$. (SPECIFICARE MEGLIO) 
%\end{remark}
%%

Now, any 1-cocycle $z$ defines representations of the observable net which are sharp excitations
of the reference one. To be precise,  given a diamond $o$ define 
\begin{equation}
\label{Bd:3a-1}
\rho^z(o)_a(A) \, := \, z(p) A z(\bar p) \ , \qquad  a\in KM,  \, A\in\cA_a \, ,  
\end{equation}
where $p : e \to o$ is an arbitrary path such that $e \subset a^\perp$.  
By \eqref{app:1}, given in appendix, this definition is independent of $p$, and $\rho^z(o)_a:\cA_a\to\cB(\cH_0)$ is a unital morphism 
satisfying the relation $\rho^z(o)_{a} = \rho^z(o)_{\tilde a}\restriction \cA_a$ for any inclusion $a\subseteq \tilde a$.
This amount to saying that the family $\rho^z(o) := \{\rho^z(o)_a\}_{a\in KM}$ defines a Hilbert space representation
\begin{equation}
\label{Bd:3a-0}
\rho^z(o) : \cA_{KM} \to \cB(\cH_0)
\end{equation}
in the sense of \cite{RV11}. 
Furthermore, always using the property \eqref{app:1}, we find that $\rho^z(o)$ is \emph{localized} in $o$, that is,
\begin{equation}
\label{Bd:3a}
 \rho^z(o)_a= \mathrm{id}_{\cA_a} \, , \qquad  a\perp o \ , 
\end{equation}
and \emph{transportable} i.e. 
\begin{equation}
\label{Bd:3b}
z(p) \rho^z(o) \ = \ \rho^z(a) z(p) \ , \qquad p:o\to a \, ,
\end{equation}
where the above equation should be read as 
\begin{equation}
\label{Bd:3b-a}
z(p) \cdot \rho^z(o)_e (A) \ = \ \rho^z(a)_e (A) \cdot z(p) \ \ \ ,\ \ \ e \in KM \, , \, A \in \cA_e \, . 
\end{equation}
Equation \eqref{Bd:3a} says that  $\rho^z(o)$ describes a charge localized in $o$, since it equals the reference representation in the causal complement of $o$. On the contrary equation \eqref{Bd:3b} illustrates the role of $z$ as a charge transporter along paths. The important point is that the charges described by AB and DHR cocycles are the same since $\rho^{z}(o)=\rho^{z_c}(o)$ for any $o\in KM$.						
This is an immediate consequence of the path independence of the definition of $\rho^z(o)$
and of the definition of the charged component $z_{\mathrm c}$. However an important difference arises when the 
charge is transported from $o$ to $a$ along to different paths $p,q:o\to a$. If $z$ is of DHR type then the transport is path-independent,
\[
z(p)\rho^z(o)z(\bar q)= \rho^z(a) \, ,
%\ \ \ {\mathrm{i.e.}} \ \ \
%z(p)\cdot \rho^z(o)_e(A) \cdot z(\bar q)= \rho^z(a)_e(A)
%\, , \, \forall e \in KM \, , \, A \in \cA_e
%\, ,
\] 
in fact $z(p*\bar q)=1$; whilst if $z$ is of AB type then 
\begin{equation}
\label{eq.Brem}
z(p)\rho^z(o)z(\bar q)= z(p*\bar q)\rho^z(a) = u_z(p*\bar q)\rho^z(a) \ .   
\end{equation}
Hence $z(p*\bar q)\rho^z(a)$ depends on the homotopy equivalence class of the loop  $p*\bar q$. In particular 
$z(p*\bar q)\ne 1$ if $p*\bar q \not\sim 1$. 
We conclude that \emph{the Aharonov-Bohm effect manifests itself in sectors of} AB \emph{type and,    
as we shall prove in the next section, it is due to a presence of a background flat potential}.

%\begin{remark}
%Let $a \in KM$. By \cite{BR08} we have that, given $p : a \to a$, the unitary $z(p)$ lies in the commutant 
%of the representation (\ref{Bd:3a-0}).
%\end{remark}

\subsection{Background flat potentials and the AB effect} 
\label{Bd}

In this section we complete the picture drawn in the preceding sections aimed at explaining the Aharonov-Bohm effect in terms of superselection sectors of local observables. In particular  we shall prove that the topological component of 1-cocycles 
are nothing but the holonomy of a flat connection.
To this end we shall use a result that has been rediscovered by several authors 
{\footnote{See, for example, \cite{Hor80} and references therein.}}.
We refer in particular to a paper by Barrett \cite{Barr}.

\medskip

Before proceeding it is convenient to specify the notion of flat connection.
Let $G$ be a compact Lie group with Lie algebra ${\EuFrak g}$ and $\cP \to M$ denote a smooth $G$-principal bundle.
We assume that $\cP$ is trivialized on any diamond, and that it is \emph{locally constant}
in the sense that it admits a set of \emph{locally constant} transition maps
\[
g_{oa} : o \cap a \to G \ \ \ , \ \ \ o,a \in KM \, .
\]
Thus $dg_{oa} \equiv 0$, and following \cite[\S I.1-2]{Kob} we define a flat connection on $\cP$ as a
collection $A$ of of ${\EuFrak g}$-valued 1--forms $A_o$ defined on $o \subset M$, fulfilling the relations
\[
A_a \restriction_{o \cap a}    \ = \    g_{oa}^* A_o g_{oa}  \restriction_{o \cap a}   \ \ \ , \ \ \ o \cap a \neq \emptyset \, .
\]
When $G = \bU(1)$, each $g_{oa}$ is a phase and $A_a \restriction_{o \cap a} = A_o \restriction_{o \cap a}$;
thus the forms $A_o$ can be glued and give a $\mathfrak{u}(1)$-valued 1--form that we denote again by $A$. 
Because of flatness, and since $\mathfrak{u}(1) \simeq \bR$, we may regard $A$ as a closed real 1--form, in symbols
\begin{equation}
\label{eq.AbConnForm}
A \in Z_{dR}^1(M) \ \ \ , \ \ \ dA = 0 \, .
\end{equation}
Now, following \cite{BR08},  if $z \in Z^1_{\mathrm{AB}}(\cA_{KM})$
is irreducible  with statistical dimension $d(z)$, the von Neumann algebra
\begin{equation}
\label{Bd:4}
\cA(M,o) \, := \, \{ z(p)=u_z(p) \ | \ p:o\to o \}'' \subseteq \cA_o
\end{equation}
is a factor of type $I_n$ with $n\leq d(z)$, where $d(z)$ is the statistical dimension of $z$. This amount to saying that 
$\cA(M,o)\cap \cA(M,o)'=\mathbb{C} 1$ and 
\begin{equation}
\label{Bd:5}
U\cA(M,o)U^* \ = \ 1_{\cH_0}\otimes\mathbb{M}_n(\mathbb{C}) \, ,
\end{equation}
for some unitary operator $U:\cH_0\to \cH_0\otimes \bC^n$.
The integer $n$ is an invariant of the equivalence class of $z$ which is called the \emph{topological dimension} of $z$ and it is denoted by $\tau(z)$.

%[ev] Vecchia versione
%\begin{lemma}
%\label{Bd:6}
%Let $z$ be an irreducible 1-cocycle of $Z^1_{\mathrm{AB}}(\cA_{KM})$ with topological dimension $n := \tau(z)$. Let $P_o$ %be a path frame. 
%Then there exists a smooth $\bU(n)$-principal bundle over $M$ with a flat connection 1-form $A^z$, taking values 
%in the Lie algebra $\mathfrak{u}(n)$,  
%such that
%%
%\begin{equation}
%\label{lemBd:6}
%u_z(\tilde a,a)
%\ = \ 
%U^* \Big( 1 \otimes \cP\exp \oint_{\ell_{(\tilde a, a)}} A^z  \Big) \, U \, \in \cU(\cA_o) \ , 
%\qquad \tilde a \gtrless a  \ , 
%%
%\end{equation}
%%
%where $U$ is the unitary realizing \eqref{Bd:5} and
%$\cP\exp \oint_{\ell_{(\tilde a, a)}} A^z  \in \bU(n)$ is the path-ordered integral over any closed curve  
%$\ell_{(\tilde a, a)} : x_o \to x_o$, $x_o \in o$, approximated by $p_{o\tilde a}*(\tilde a,a)*p_{ao}$. 
%\end{lemma}
%
%
%
%[ev] Nuova versione
\begin{lemma}
\label{Bd:6}
Let $z$ be an irreducible 1-cocycle of $Z^1_{\mathrm{AB}}(\cA_{KM})$ with topological dimension $n := \tau(z)$. 
Let $P_o$ be a path frame. 
Then there exists a smooth $\bU(n)$-principal bundle $\cP^z \to M$ with a flat connection $A^z$  such that
\begin{equation}
\label{lemBd:6}
u_z(\tilde a,a)
\ = \ 
U^* \Big( 1 \otimes {\mathrm{hol}}_{A^z} (\ell_{(\tilde a, a)}) \Big) \, U \, \in \cU(\cA_o) \ , 
\qquad \tilde a \gtrless a  \ , 
\end{equation}
where $U$ is the unitary realizing \eqref{Bd:5} and
${\mathrm{hol}}_{A^z} (\ell_{(\tilde a, a)}) \in \bU(n)$ is the holonomy over any closed curve  
$\ell_{(\tilde a, a)} : x_o \to x_o$, $x_o \in o$, approximated by $p_{o\tilde a}*(\tilde a,a)*p_{ao}$. 
\end{lemma}

\begin{proof}
As observed in the previous section $z$ and its topological component define the same representation $\si_z$  \eqref{Bc:1} of the fundamental group of $M$. As the algebra $\cA(M,o)$ is a type $I_{n}$ factor, there exists an irreducible $n$-dimensional representation $\si$ of the fundamental group of $M$ such that 
\begin{equation}
\label{Bd:7}
\si_z([\ell]) \ = \ u_z(p_\ell)= U^*(1_{\cH_0}\otimes\si([\ell]))U \ ,\qquad [\ell]\in\pi_1(M,x_o) \ , 
\end{equation}
where $U$ is the unitary realizing \eqref{Bd:5}. The key observation is that the  
mapping $\ell \mapsto\sigma([\ell])$ assigning to any closed curve $\ell$ over $x_o$ 
the unitary $\sigma([\ell]) \in \bU(n)$ is a holonomy map in the sense of Barrett \cite[cfr. Section 2.4]{Barr}. 
By Barrett's reconstruction theorem \cite{Barr}, $\ell \mapsto\si([\ell])$ is the parallel transport of a flat connection $A^z$ on a smooth $\bU(n)$-principal bundle $\cP^z$, that is,
$\si([\ell]) \ = \ \mathrm{hol}_{A^z} (\ell)$ for any closed curve $\ell:x_e\to x_e$.
%
%(see, for instance, \cite[\S II]{BM}).
%
Explicitly,
\begin{equation}
\label{eq.Pz}
\cP^z \, := \, \hat{M} \times_\sigma \bU(n) \, ,
\end{equation}
where $\hat{M}$ is the universal covering of $M$ carrying the well-known right $\pi_1(M)$-action, 
whilst $\hat{M} \times_\sigma \bU(n)$ is the quotient of $\hat{M} \times \bU(n)$ by the equivalence relation
$(y\ell,v) \simeq (y,\sigma(\ell)v)$, $y \in \hat{M}$, $\ell \in \pi_1(M)$, $v \in \bU(n)$.
Hence by \eqref{Bc:1} and \eqref{Bd:7} we have that 
\begin{equation}
\label{Bd:8}
u_z(p_\ell) \ = \ U\,\Big(\bI \otimes {\mathrm{hol}}_{A^z} (\ell) \Big)\,U^* \ , 
\end{equation}
for any loop $p_\ell$ and any closed curve $\ell$ such that $p_\ell$ is a 
path approximation of $\ell$. Given a comparable pair $\tilde a \gtrless a$, if we take the loop
$p_{o\tilde a}*(\tilde a,a)*p_{ao}$ and a closed curve $\ell_{(\tilde a,a)}:x_o\to x_o$  
such that $p_{o\tilde a}*(\tilde a,a)*p_{ao}$ is a path-approximation of $\ell_{(\tilde a,a)}$,  
the proof follows from \eqref{Bd:7}, \eqref{Bd:8} and from the definition of the topological component. 
\end{proof}

It is easily seen that a change of the path frame $P_o$ used in the previous proof 
yields a representation $\si' : \pi_1(M,x_{o'}) \to \bU(n)$ unitarily equivalent to $\si$.
This implies that we find a flat connection ${A'}^z$ providing, up to conjugation, 
the same parallel transport as $A^z$.
\begin{remark}[The Chern character on $Z^1_\rAB(\cA_{KM})$]
\label{rem.Chern}
Let $z$ denote an irreducible cocycle and let $\sigma : \pi_1(M) \to \bU(n)$, $n:=\tau(z)$, denote the associated representation of the fundamental group (\ref{Bd:7}).
By standard results in differential geometry, $\sigma$ defines a flat Hermitean vector bundle $E^z \to M$ 
with rank $n$, 
\begin{equation}
\label{eq.Es}
E^z \, := \, \hat{M} \times_\sigma \bC^n \, ,
\end{equation}
defined analogously to (\ref{eq.Pz}).
Using the results in \cite{CS85}, we can define the Chern character 
\begin{equation}
\label{eq.ccs}
ccs : Z^1_\rAB(\cA_{KM}) \to H^{odd}(M,\bR/\bQ)
\ \ \ , \ \ \ 
ccs(z) \, := \, n + \sum_k^n \frac{ (-1)^{k-1} }{ (k-1)! } c_k(E^z)_{ {\mathrm{mod}} \bQ }
\, ,
\end{equation}
defined via the Chern-Cheeger-Simons classes $c_k(E^z) \in H^{2k-1}(M,\bR / \bZ)$ \cite[\S 4]{CS85}.
The interpretation of $ccs(z)$ is that of a generalized Aharonov-Bohm phase, as we shall see in the next sections.
%
%for, if in particular $\tau(z) = 1$ then by definition
%%
%\[
%\{ c_1(E_\si) \}(\ell) \, = \, \exp \oint_\ell A^z   \, ,
%\]
%%
%having identified $\bR/\bZ$ with $\bU(1)$ and regarded the loop $\ell$ as a 1-cycle in $Z_1(M)$.
%In the general case $\tau(z) \geq 1$, each component $c_k(E_\si)$ yields a phase assigned to $(2k-1)$-cycles.
%
%
\end{remark}

\medskip

\noindent From now on, to  emphasize the geometric aspects, we shall write 
\[
{\mathrm{Hol}}_{A^z}(\tilde a, a)  \ := \  u_z(\tilde a,a) \ ,  \qquad 
%\Big\{ \cP\exp\oint A^z \Big \} (\tilde a,a) \, := \, u_z(\tilde a,a) \, \in \cA(M,o) 
 \tilde a \gtrless a \, ,
\]
instead of the expression appearing in (\ref{lemBd:6}).
We can state the main result of this section. 
\begin{theorem}
\label{Bd:9}
Let $z$ be an irreducible 1-cocycle of $Z^1_{\mathrm{AB}}(\cA_{KM})$ with $d(z)=\tau(z)$. Then 
the charged component $z_{\mathrm{c}}$ is a multiplet of the type $z_{\mathrm{c}}  = \vec{\zeta}$,
where $\vec{\zeta} := \zeta \oplus \ldots \oplus \zeta$ is a $\tau(z)$-fold direct sum
of a topologically trivial 1-cocycle $\zeta$ such that $d(\zeta)=1$ and $\kappa(\zeta) = \kappa(z)$. 
Hence  
\[
z(\tilde a,a) \ = \ \Big\{ {\mathrm{Hol}}_{A^z}  \Join \vec{\zeta} \Big\} (\tilde a,a) \ ,
\qquad  \tilde a \gtrless a \ ,  
\]
where $A^z$ is the flat connection of the  $\bU(n)$-principal bundle $\cP^z$ defined by $z$. 
\end{theorem}
\begin{proof}
In general, since $\sum^m_{k=1} d(\zeta_k)= d(\tau)$, 
the charged component can be written as a direct sum of $m$ irreducible
and topologically trivial 1-cocycles $z_{\mathrm{c}}=\oplus^m_{k=1} \zeta_k$, with $m\leq d(z)$ \cite{BR08}.
We prove that $m=d(z)$. First of all observe that 
\[
\cA(M,e) \, \subseteq \, (z,z)_e \, := \, \{ t_e \ , \ t\in (z,z) \} \, ,
\] 
where the r.h.s. is a $\rC^*$-algebra formed by the component $e$ of the intertwiners of $(z,z)$.
If $m<d(z)$, then the dimension of $(z,z)_e$ would be smaller that $d(z)^2$ 
because any $\zeta_k$ is irreducible. This leads to a contradiction because of the above inclusion 
and because the dimension of $\cA(M,e)$ is $\tau(z)^2$ and $\tau(z)=d(z)$ by assumption. 
This implies that all $\zeta_k$ have statistical dimension $1$ and,
on account of the same dimensional argument just used, that all $\zeta_k$ are equivalent. 
Thus the proof follows by applying the previous theorem to \eqref{Bd:3}.
\end{proof}
\noindent We shall refer to $A^z$ as the \emph{background flat potential} associated with $z$. 

\medskip

Hence, an irreducible 1-cocycle $z$ with statistical dimension $d(z) = \tau(z)$ is the transporter of 
charge of Fermionic/Bosonic type in the presence of a background field $A^z$. 
If, in particular, $d(z)>1$ then we see that the charge $z_{\mathrm{c}}$ is a multiplet,
and the  background field acting upon this multiplet, realizing the AB effect, is non-Abelian.
The background field quantifies how the charge depends on the homotopy equivalence of paths when transported along paths. 
Precisely, referring to (\ref{eq.Brem}), if $p:o\to a$ and $q:o\to a$ then 
\[
z(p)\, \rho^z(o) \, z(\bar q) 
\ = \ 
{\mathrm{Hol}}_{A^z} (p*\bar q) \ \rho^z(a) \, ,
\]  
so the final representation $ \rho^z(o)$ differs from $\rho^z(a)$ by the factor
${\mathrm{Hol}}_{A^z} (p*\bar q)$ (see \eqref{Bd:3b-a} for the notation).

The next result proves that non-Abelian AB effects can appear only in spacetimes having 
a non-Abelian fundamental group.
\begin{corollary}
\label{Bd:9a}
Let $z \in Z^1_{\mathrm{AB}}(\cA_{KM})$ be irreducible  with topological dimension $\tau(z) = 1$.
Then ${\mathrm{Hol}}_{A^z}$ takes values in $\bU(1)$ and
\begin{equation}
\label{eq.Bd:9a}
z(\tilde a,a) \ = \ 
{\mathrm{hol}}_{A^z}(\ell_{(\tilde a,a)})  \cdot \, z_{\mathrm{c}}(\tilde a,a) 
\ , \qquad \forall \tilde a \gtrless a \, .
\end{equation}
When $\pi_1(M)$ is Abelian, this holds for any irreducible 1-cocycle.
\end{corollary}
\begin{proof}
It is enough to observe that if $\tau(z)=1$  the topological component takes values in $\bU(1)$
and the joining reduces to the multiplication by the so-defined phase. 
Moreover, if $\pi_1(M)$ is Abelian, then any irreducible representation of $\pi_1(M)$ is a phase.
\end{proof}
%
%\Big( \exp \oint_{\ell_{(\tilde a,a)}}\hspace{-0,1cm} A^z \Big) \cdot \, z_{\mathrm{c}}(\tilde a,a) 

%[ev]
%
%For details on the next standard results, we recommend \cite[\S I.1-2]{Kob} .
%Let $z \in Z^1_{\mathrm{AB}}(\cA_{KM})$ and $\{ g_{oa} \}$ denote a set of $\bU(n)$-valued transition maps for $\cP^z$,
%defined on intersections of diamonds $o,a$. 
%Then $g_{oa}$ are \red{locally} constant, and the connection form defined by $D^z$ 
%is a collection $A^z$ of $\mathfrak{u}(n)$-valued 1--forms $A^z_o$ defined on $o \subset M$, fulfilling the relations
%%
%\[
%A^z_a \restriction_{o \cap a}    \ = \    g_{oa}^* A^z_o g_{oa}  \restriction_{o \cap a}   \ \ \ , \ \ \ o \cap a \neq \emptyset \, .
%\]
%%
%We call $A^z$ the \emph{flat gauge potential defined by $z$}. When $n=1$, each $g_{oa}$ is a phase and $A^z_a \restriction_{o \cap a} = A^z_o \restriction_{o \cap a}$;
%thus we can glue the $A^z_o$'s and get a $\mathfrak{u}(1)$-valued 1--form that we denote again by $A^z$. 
%Since $D^z$ is flat, and since $\mathfrak{u}(1) \simeq \bR$, we may regard $A^z$ as a closed real 1--form, in symbols
%%
%\begin{equation}
%\label{eq.AbConnForm}
%A^z \in Z_{dR}^1(M) \ \ \ , \ \ \ dA^z = 0 \, .
%\end{equation}
%%
%Now, it is known that a flat principal $\bU(1)$-bundle $\cP$ with monodromy
%%
%$\sigma : \pi_1(M) \to \bU(1)$ 
%%
%is trivial if, and only if, 
It is known that a principal $\bU(1)$-bundle $\cP$ with a flat connection $A$
is trivial if, and only if, 
\begin{equation}
\label{eq.expA}
\mathrm{hol}_{A}([\ell]) \ = \ \exp \Big(i \oint_\ell A\Big) \ \ \ , \ \ \ \forall \ell \, . 
\end{equation}
That such a $\cP$ is trivial shall be proved in the following \S \ref{sec.DA1}.
The converse can be deduced from \cite[Cor.1.2]{CS85}: if $\cP$ is trivial, then its Chern class
$\delta_2(\cP) \in H^2(M)$ vanishes\footnote{
To fit the formalism of \cite[Cor.1.2]{CS85}, we used the symbol $\delta_2$ to denote the \emph{first} Chern class. 
For the same reason, $\mathrm{hol}_{A}$ should be regarded as a cohomology class in $H^1(M,\bR/\bZ)$ having identified $\bU(1)$ with $\bR/\bZ$. Finally, $H^1(M,\bR/\bZ)$ and $H^2(M)$ stand for the 
singular cohomology groups, whilst $H_1(M)$ is the first homology group.}.
Hence by the second exact sequence in \cite[Theorem 1.1]{CS85}, $\mathrm{hol}_{A}$ must be of the type (\ref{eq.expA}).
Finally, as $\cP$ is flat,  $\delta_2(\cP)$ is either zero or is a torsion class, that is, 
there is $k \in \bN$ such that $k \delta_2(\cP) = 0$ \cite[Sec.3]{Fre86}.
Thus applying Corollary \ref{Bd:9a} we conclude:
\begin{corollary}
\label{Bd:9a2}
Let $z \in Z^1_{\mathrm{AB}}(\cA_{KM})$ be irreducible with  $\tau(z) = 1$, and 
$\sigma : \pi_1(M) \to \bU(1)$ 
denote the morphism (\ref{Bd:7}). If $\delta_2(\cP^z) = 0$, then
\begin{equation}
\label{eq.Bd:9a2}
z(\tilde a,a) \ = \ \exp \Big(i\oint_{\ell_{(\tilde a,a)}}A^z\Big) \cdot \, z_{\mathrm{c}}(\tilde a,a) 
\ , \qquad  \tilde a \gtrless a \, .
\end{equation}
This always holds true when  $H^2(M)$ has no torsion classes.
\end{corollary}

By well-known results, the above condition of $H^2(M)$ being torsion free is equivalent
to $H_1(M)$ being torsion free (see \cite[Sec.3]{Fre86} and references cited therein).
This covers the cases of physical interest, 
such as anti-de Sitter spaces ($H_1(M) = \bZ$), de Sitter spaces ($H_1(M) = 0$),
and the "Aharanov-Bohm spacetime" $M := (\bR^3 \setminus S) \times \bR$, 
where $S \sim \bR$ stands for the ideally infinite solenoid ($H_1(M) = \bZ$).

\section{On the observable net of the free Dirac field}
\label{C}

In the previous section we have shown that the nontrivial topology of the spacetime induces on the observable net 
superselection sectors affected by the topology of the spacetime. These sectors factorize in a part which describes a quantum charge (DHR charges) and in a topological part ruling the behaviour of the quantum charge when moved along a path in the same way as it occurs in the Aharonov-Bohm effect.  

Our final aim is to show that this is not merely an analogy but the real physical Aharonov-Bohm effect: we want to show that the quantization of the free Dirac field in the presence of a  background flat potential gives rise to the same type of superselection sectors as those described in the abstract analysis of the previous section.

As a preliminary step,  in this section we analyze the free Dirac field
in a 4-d globally hyperbolic spacetime \emph{without} the presence of a background flat potential. We 
show that the net of local observables  in the representation defined by a pure quasi-free Hadamard state satisfies all the properties assumed in section \ref{Ba},  in particular Haag duality and (a strengthening)  
of the Borchers property. Then we provide an explicit construction 
of the superselection sectors of the Dirac field and we show, as expected, that these are DHR cocycles obeying the Fermi statistics.

\subsection{The Dirac operator}

It is convenient to recall some notions on spin structures. The existence of any such structure is guaranteed if and only the second Stiefel-Whitney class of the underlying manifold $w_2(M)\in H^2(M,\bZ_2)$ vanishes. This is always the case if we consider four dimensional globally hyperbolic spacetimes and, in addition the corresponding Dirac bundle is trivial, namely $DM\simeq M\times\bC^4$, {\it cf.} \cite{Ish78} and references therein. We introduce also
the space $\cE(DM)$ of smooth sections of $DM$,
the space $\cS(DM)$ of \emph{compactly supported} smooth sections of $DM$ (\emph{spinors}),
and the subspaces $\cS_o(DM)$ of those sections $f\in \cS(DM)$ whose support is contained in a diamond $o$,
$\mathrm{supp}(f)\subset o$. 
%
% Associated to the spin structure we also have a Clifford bundle $CM \to M$, 
% with fibre the Clifford algebra generated by Dirac matrices 
% %
% $\gamma^\mu$, $\mu=0,\ldots,3$,
% %
% depending on $x \in M$ because of the metric. 
% %
% The spin connection 
% %
% \[
% \nabla : \cS(DM) \to \cS( T^*M \otimes {\bf end}DM ) \ ,
% \]
% %
% where $T^*M$ is the cotangent bundle and ${\bf end}DM$ is the bundle of endomorphisms of $DM$, 
% is then written in local frames as
% %
% \[
% \nabla_\nu \ = \  \partial_\nu + \frac{1}{4} \, \Gamma_{\mu\nu}^\rho \gamma_\rho \gamma^\mu  \, ,
% \]
% %
% where $\Gamma_{\mu\nu}^\rho$ are Christoffel symbols,
% and characterized by the property $\nabla_\mu \gamma_\nu = 0$. 
% We can then define the Dirac operator 
% %
% $i\slashed\nabla : \cS(DM) \to \cS(DM)$
% %
% by contraction by Dirac matrices, in local frames $i\slashed\nabla := \gamma^\mu_o \nabla_\mu$.
% Here, 
% %
% \begin{equation}
% \label{eq.C1}
% \gamma^\mu_o \ \ \ , \ \ \ \mu=0,\ldots,3 \, ,
% \end{equation}
% %
% are local Dirac matrices generating the fibres of the Clifford bundle over $o$,
% associated with the choice of a local frame for $T^*M$.
%
%
Similar notions hold for the dual Dirac bundle $D^*M\cong M\times (\mathbb{C}^4)^*$.   
We denote the space of smooth sections by $\cE(D^*M)$, that of compactly supported smooth sections (co-spinors) by  $\cS(D^*M)$ and the subset of smooth sections of $\cS(D^*M)$
whose support is contained in a diamond $o$ denoted by $\cS^*_o(DM)$.

On top of these structures we define the Dirac operator $\mathcal{D}:= i\gamma^\mu\nabla_\mu-m$ acting on spinors, where $m>0$ is the mass parameter, the $\gamma$-matrices are fiberwise linear maps of $DM$ obeying $\{\gamma_a,\gamma_b\}=2g_{ab}$, while $\nabla$ is the spin-connection. Similarly we define the dual Dirac operator $\mathcal{D}_*=-i\gamma^\mu\nabla_\mu-m$ acting on co-spinors. Here duality refers to the metric induced pairing 
$\langle,\rangle:\cE(D^*M)\times \cS(DM)\to {\mathbb{C}}$ 
\[
\langle h,f\rangle :=\int\limits_M  h\cdot f  \, d\mu_g  , \qquad 
h\in \cE(D^*M), \ f\in\cS(DM) \ . 
\]
Since the Dirac operator is Green hyperbolic it possesses unique advanced and retarded fundamental solutions
$S^\pm:\cS(DM)\to \cE(DM)$,  such that $\mathcal{D}S^\pm=S^\pm\mathcal{D}=id$ on $\cS(DM)$ while $\textrm{supp}(S^\pm(f))\subseteq J^\pm(\textrm{supp}(f))$ for all $f\in \cS(DM)$. The same conclusion can be drawn for $\mathcal{D}_*$ and it turns out that the associated advanced and retarded fundamental solutions $S^{\pm}_*:\cS(D^*M)\to \cE(D^*M)$ obey the structurally defining property 
$\langle h, S^\pm(f)\rangle=\langle S^{\mp}_*(h), f\rangle$ for any $h\in\cS(D^*M)$ and $f\in\cS(DM)$. Finally the corresponding 
propagators  $S:=S^+- S^-$ and $S_*:=S^{+}_*- S^{-}_*$ enjoy 
\[
\langle S_*(h),f\rangle=-\langle h,S(f)\rangle \ , \qquad h\in \cS(D^*M), \ f\in\cS(DM) \ .
\] 
The vector spaces whose elements are spinors and co-spinors are isomorphic via the Dirac adjoint  defined in terms of the gamma matrices  as
\[
\cE(DM)\ni f\mapsto f^\dagger:= f^*\, \gamma_0\in \cE(D^*M)  \ \ \ ; \  \ \ \cE(D^*M)\ni h\mapsto h^\dagger:= \gamma^{-1}_0 h^*\in 
\cE(DM) \ .
\]
Here we are taking the gamma matrices in the standard representation: $\gamma^*_0=\gamma_0$ and $\gamma^*_k=-\gamma_k$  for 
$k=1,2,3$.

\subsection{The (CAR) algebra of the Dirac field}
\label{Ca}

There is a fairly vast literature on the $\textrm{CAR}$ algebra and the Dirac quantum field in a 4-dimensional globally hyperbolic spacetime. 
In the present section we give a brief description of these topics, and refer the reader to the references \cite{Dim82,San10,DHP09} for details. 

\medskip

A convenient approach to quantize the Dirac field and to study the corresponding representations is the self-dual approach \cite{Ara71}.  We consider the following Whitney sum of bundles 
\[
\widetilde{D}M:= DM\oplus D^*M \ , 
\]
as well as the space of compactly supported smooth sections $\cS(\widetilde{D}M)\cong \cS(DM)\oplus\cS(D^*M)$.  
In addition we extend the Dirac operators and the corresponding propagators on this space by setting 
\[
\widetilde D (f\oplus h):= Df\oplus D_*h \ , \qquad f\oplus h \in \cS(\widetilde{D}M).
\]
Observe that $\cS(\widetilde{D}M)$ is equipped with a positive semi-definite sesquilinear form 
\begin{equation}
\label{scalar}
(f_1\oplus h_1, f_2\oplus h_2):=  (f_1,f_2)_s+ (h_1,h_2)_c \ ,  
\end{equation}
with  
\begin{equation}
\label{Scalars}
(f_1,f_2)_s:= i\int_M f^\dagger_1 S(f_2)   \ \ ; \ \ 
(h_1,h_2)_c:=  -i\int_M S_*(h_2)  h^\dagger_1 \ , 
\end{equation}
for any  $f_i\oplus h_i\in\cS(\widetilde{D}M)$,  $i=1,2$ (both $(\cdot,\cdot)_s$ and $(\cdot,\cdot)_c$ 
are positive semi-definite sesquilinear forms on $\cS(DM)$ and $\cS(D^*M)$ respectively). 
This form \eqref{scalar} annihilates on the image of the extended Dirac operator 
$\widetilde D$. Taking the quotient we get, up to a suitable closure, a Hilbert space 
(${\bs h}, (,))$. Finally setting
\[
\Gamma(f\oplus h):= h^\dagger  \oplus f^\dagger \ , \qquad f\oplus h\in S(\widetilde{D}M)
\]
one observes that $\Gamma$ extends to an anti-unitary operators on $\bs h$:
\[
(\Gamma [f_1\oplus h_1], \Gamma [f_2\oplus h_2])=([f_2\oplus h_2],[f_1\oplus h_1]) \ , \qquad  
[f_1\oplus h_1],[f_2\oplus h_2]\in\bs h. 
\]
% Finally note that for any phae $e^{i\varphi}$ we have 
% \[
% f\oplus g\mapsto fe^{i\varphi}\oplus he^{-i\varphi} 
% \] 
% leaves invariant the scalar product and commutes with the involution $\Gamma$.
%  
We now are ready to give the definition of the CAR algebra. 
\begin{definition}
\label{CAR}
The \textbf{algebra of the Dirac field} is the $\mathrm{C}^*$-algebra $\cC(M)$ generated by the unit $\mathbbm{1}$
and $\Psi([f\oplus h])$ as $[f\oplus h]$ varies on the Hilbert space ${\bs h}$ and satisfying 
\begin{itemize}
 \item[(i)] The map $\bs h\ni [f\oplus h]\mapsto \Psi([f\oplus h]) \in \cC(M)$ is linear;
 \item[(ii)] Self-duality: $\Psi( \Gamma([f\oplus h])) =\Psi([f\oplus h])^*$ for any $[f\oplus h]\in \bs h$;
 \item[(iii)] Canonical Anticommutation Relations (CARs): 
 \[
 \{\Psi([f_1\oplus h_1]),\Psi([f_1\oplus h_2])\} = 
 \left( \Gamma [f_1\oplus h_1],[f_2\oplus h_2]\right)
 \]
 for any $ [f_1\oplus h_1],[f_2\oplus h_2]\in \bs h$.
\end{itemize}
\end{definition}
\noindent We draw on the consequences of this definition.\smallskip  

$(1)$ The generators of the CAR algebra encode \emph{the Dirac field $\psi$ and the conjugated Dirac field $\bar{\psi}$}:
\[
\psi (f):= \Psi[f\oplus 0] \, , \ \ f\in\cS(DM) \ \ \ \ ; \ \ \ \  
\bar{\psi}(h):= \Psi[0\oplus h] \,,  \ \ h\in\cS(D^*M) \ . 
\]
Notice in fact that, since $\Psi\circ \widetilde D=0$, the fields  $\psi$ and $\bar{\psi}$ satisfy the  \emph{Dirac equation} 
\[
\psi\circ (i\slashed \nabla-m)=0 \ \ , \ \  \bar{\psi}\circ (i\slashed \nabla+m) =0 \ . 
\]
Furthermore we have  
\begin{equation}
\label{dircon}
\bar{\psi}(h)= \psi(h^\dagger)^* \ , \qquad h\in \cS(D^*M) \ , 
\end{equation}
and thus  $\Psi[f\oplus h]= \psi(f)+\psi(h^\dagger)^*= \bar{\psi}(f^\dagger)^*+\bar{\psi}(h)$. This implies in particular that the algebra $\cC(M)$ can be generated by taking only either the Dirac field or the conjugated one. \smallskip

(2) The Dirac field satisfies the usual   \emph{Canonical Anticommutation Relations}:
\begin{equation}
\label{CARs}
\{ \psi(f_1)^* , {\psi}(f_2) \} \ = \ (f_1,f_2)_s 
\ \ \ , \ \ \  
\{ {\psi}(f_1) , {\psi}(f_2)\} \ = \ 0 \ .
\end{equation}
From these and \eqref{dircon} it follows that $\{ \bar{\psi}(h) , {\psi}(f) \} = (h^\dagger,f)_s$ and 
$\{\bar{\psi}(h_1),\bar{\psi}(h_2)\}=0$. \smallskip 

(3) The algebra $\cC(M)$ is acted upon by $\bU(1)$ as a \emph{gauge group}. The mapping 
$[f\oplus h] \mapsto [\zeta f\oplus \bar\zeta h ]$, which leaves invariant the scalar product and commutes 
with the conjugation $\Gamma$ for any $\zeta\in\bU(1)$, induces 
an action $\bU(1)\ni\zeta \mapsto \eta_\zeta \in\mathrm{Aut}(\cC(M))$ where 
\begin{equation}
\label{gauge0}
\eta_\zeta( \Psi([f\oplus h]))= \Psi(\zeta f\oplus \bar{\zeta} h) \ .
\end{equation}
In particular $\eta_\zeta(\psi(f)= \zeta \psi(f)$ and 
$\eta_\zeta(\bar{\psi}(h)= \bar{\zeta} \, \bar{\psi}(h)$. We denote by $\eta_-$ the automorphism induced by 
$\zeta=-1$. \smallskip 

(4) Finally,  we note that 
$\psi$ is a $\mathrm{C}^*$-valued distribution when  $\cS(DM)$ is equipped with the topology of test-functions (the same holds true for $\bar{\psi}$). In fact by the CARs \eqref{CARs} we have $\psi(f)^*\psi(f)= - \psi(f)\psi(f)^* +  (f,f)_s\mathbbm{1}$. So  
$\psi(f)^*\psi(f)\psi(f)^*\psi(f)= \psi(f)^*\psi(f) \, (f,f)_s\mathbbm{1}$ 
because $\psi(f)\psi(f)=0$. Using the $\mathrm{C}^*$ property of the norm
we get 
\begin{multline*}
\|\psi(f)^*\psi(f)\|^2=\|\psi(f)^*\psi(f)\psi(f)^*\psi(f)\| = \|\psi(f)^*\psi(f)\| (f,f)_s  \\ 
  \iff (f,f)_s =\|\psi(f)^*\psi(f)\|= \|\psi(f)\|^2 \ ,
\end{multline*}
and the proof follows by the continuity of $(\cdot,\cdot)_s$ with respect to the test function topology.

\subsection{Quasi free Hadamard states}

We describe the net of the Dirac field in the representation associated with a pure and gauge invariant quasi-free Hadamard state. We show that the net satisfies twisted Haag duality and the local algebras are type $\mathrm{III}$ factors. As a consequence the resulting net of local observables in the vacuum (zero charge) representation verifies all the properties assumed in Section \ref{Ba}.\smallskip 
 
To begin with we consider a \emph{pure quasi-free state} $\omega$ 
of the algebra $\cC(M)$ and the corresponding  GNS triple $(\cH,\pi,\Omega)$ (see \cite{Ara71}). Here $\cH$ is the antisymmetric Fock space associated with 
a closed subspace of $P\bs h$  defined by an orthogonal projection $P$, \emph{the base projection} enjoying 
the relation $\Gamma P=(1-P)\Gamma$; $\pi$ is an irreducible  representation of the CAR algebra on the 
Hilbert space $\cH$ and  $\Omega$ is a unit norm vector $\cH$ 
such that $\omega(\cdot)=(\Omega, \pi(\cdot) \Omega)$.  To ease notation we keep symbols 
$\Psi$, $\psi$ and $\bar{\psi}$ to denote the generator of the CAR algebra, the Dirac field and its conjugate, respectively, in the Fock representation. Finally, if $\omega$ is also 
\emph{gauge invariant} $\omega\circ\eta = \omega$ then there exists a unitary representation $U: \bU(1) \to \cU(\cH)$ such that for any $\zeta\in \bU(1)$, it holds $U(\zeta) \Omega = \Omega$
\begin{equation}
\label{gauge}
U(\zeta) \Psi([f\oplus h]) \, U(\zeta)^*  = \Psi([\zeta f\oplus \bar{\zeta}h]).
\end{equation}
In addition  $U(\zeta) \psi(f) \, U(\zeta)^*= \zeta \psi(f)$ and $U(\zeta) \bar\psi(h) \, U(\zeta)^*= \bar{\zeta} \, \bar{\psi}(h)$.\smallskip

%To ease notation, we denote the Dirac field, and its coniugate, represented on the Hilbert space by the same symbol %$\psi$, and $\bar\psi$, used for the abstract fields. Clearly they satisfies the Dirac equation \eqref{} and the %CARs \eqref{}. 

% Let $\omega$ be a pure quasi-free Hadamard state of the algebra of the Dirac field and $(\cH,\pi,\Omega)$ be the corresponding GNS triple. Briefly, $\pi$ is a representation of $\cC(M)$ on the Hilbert space $\cH$ which is the anti-symmetric Fock 
% space associated with a subspace of $\bs h$ (see Appendix); $\Omega$ is a cyclic vector of $\cH$ implementing 
% $\omega$, $\omega()=(\Omega,\pi()\Omega)$. To ease notation, we denote the Dirac field, and its coniugate, represented on the Hilbert space by the same symbol $\psi$, and $\bar\psi$, used for the abstract fields. Clearly they satisfies the Dirac equation \eqref{} and the CARs \eqref{}. 
%Since $\omega$ is quasi-free, it is $\alpha$-invariant, thus there is a unitary $\bU(1)$-covariant representation
%%
%\begin{equation}
%\label{eq.Ca4}
%U : \bU(1) \to \cU(\cH) \ \ \ , \ \ \ U_\zeta \Psi(f) \, U_\zeta^* \ = \ \zeta \psi(f)
%\ \ \  , \ \ \ 
%U_\zeta \Omega = \Omega
%\ \ \ , \ \ \ 
%\forall \zeta \in \bU(1)
%\ .
%\end{equation}
%%
%In particular $U_{-}$ be the unitary implementing the automorphism $\alpha_-$. \smallskip  
%
%We shall denote the Dirac ....
%
%
%

So given a pure and gauge invariant quasi-free state of the CAR algebra and   $\cH,\pi,\omega, U$ as above,  
the \emph{net of local Dirac fields} is the correspondence $\cF: KM\ni o \to \cF_o\subseteq \cB(\cH)$ 
associating the von Neumann algebra 
\[
\cF_o := \{\Psi([f\oplus h]) \ : \  [f\oplus h]\in\widetilde{\cS}_o(DM)\}^{\prime\prime}\  ,  
\]
to any diamond $o$. The net $\cF_{KM}$ satisfies isotony, $\cF_o \subseteq \cF_{\tilde o}$ for any 
$o \subseteq \tilde o$, but \emph{not} causality 
because of the CARs. This is replaced by \emph{twisted causality}  
\[
 \cF^t_o\subseteq \cF^\prime_a \ , \qquad a\perp o \ . 
\]
where $\cF^t_o$ is the von Neumann algebra defined as $\cF_o^t:=\tilde U_- \cF_o \tilde U^*_-$  in terms  
of  the \emph{twist operator} $\tilde{U}_-$:
$$\tilde{U}_-  := \frac{1}{1+i} \, (1+iU_-)\quad\textrm{such that}\quad\tilde{U}_- \Psi(\cdot) \tilde{U}_-^*  =  iU_-\Psi(\cdot).$$ 
\begin{theorem}
\label{ThtHaag}
The net of the Dirac field $\cF_{KM}$ in a Fock representation induced by a pure and gauge invariant quasi-free state satisfies \textbf{twisted Haag duality} 
\begin{equation}
\label{tHaag}
 \cF^t_o = \cap \{\cF^\prime_a \ : \  a\perp o\} \ , \qquad o\in KM \ . 
\end{equation}
\end{theorem}
\begin{proof}
The proof relies on abstract twisted duality and on a density result.  We recall that 
for the CAR algebra in a Fock representation, abstract twisted duality holds  
for any closed subspace which is invariant under the conjugation defining the abstract CAR algebra \cite{BJL02}.  
Observing that the closure $\bs h_o$ of  $\widetilde \cS_o(DM)$ in $\bs h$ 
is $\Gamma$ invariant for any $o$, abstract twisted duality in our case reads  as 
\[
\cF^t_o= \cF\big((\bs h_o)^c)^\prime \ , 
\]
where $(\bs h_o)^c$ is the orthogonal complement of $\bs h_o$. So, it is enough to prove 
$\cup \{\widetilde{\cS}_{a}(DM) \,:\, a\subset o^\perp\}$ is dense in $(\bs h_o)^c$  where $o^\perp$  is  the causal complement 
of the diamond $o$.\smallskip  

We now need two observations. First,  any diamond $o$ is of the form $D(G)$  with 
$G\subset \Sigma$ where $\Sigma$ is a spacelike Cauchy surface and $G$ is an  open relatively compact 
subset of $G$ diffeomorphic to an open 3-ball with $cl(G)\subsetneq \Sigma$ . So, it is easily seen that 
$o^\perp= I(\Sigma\setminus cl(G))$ where $I$ stands for the chronological set. Secondly, 
the Hilbert space $\bs h$ is isomorphic to the Hilbert space 
$L^2(\Sigma,\mathbb{C}^4)\oplus L^2(\Sigma,(\mathbb{C}^{4})^*)$ 
with the scalar product $(u_1\oplus v_1,u_2\oplus v_2):=
\int_\Sigma  (u^\dagger_1 \slashed n u_2 + v_2\slashed n v^\dagger_1)\, d\Sigma$ where $n^\mu$ is the timelike future pointing normal vector to the Cauchy surface. In particular
the isomorphism relies on the following identity
\[
([f_1\oplus h_1],[f_2\oplus h_2])= \int_\Sigma \left((Sf_1)^\dagger\slashed n (Sf_2) + (S_*h_2)\slashed n (S_*h_1)^\dagger\right) 
d\Sigma   \ . 
\]  
On these grounds it is enough to prove that the space 
$C^\infty_0(G,\mathbb{C}^4)\oplus C^\infty_0(\Sigma\setminus cl(G),\mathbb{C}^4)$ is dense in  $L^2(\Sigma,\mathbb{C}^4)$ and similarly for the dual. This follows from a partition of unity argument, similarly to \cite[Prop. III.1]{HdA06} for Majorana fields, because $G$ is a relatively compact open subset of $\Sigma$ with  a smooth boundary. 
\end{proof}
A class of quasi-free states for quantum fields on curved 
spacetimes  having several physically meaningful applications is that of Hadamard states \cite{KW91} (see also \cite{FV13}). The existence of pure quasi-free and gauge invariant Hadamard states 
has been shown in \cite{FL15}. The local algebras of the Dirac field 
in a representation defined by such a state are type $\mathrm{III}$ factors \cite{HdA06} \footnote{Actually the authors give a stronger characterization of the local algebras showing that these  
are $\mathrm{III}_1$ factors.}. This property is stronger 
than the Borchers property since for any $o\in KM$ any projection $E\in\cF_o$ is equivalent in $\cF_o$ to the identity i.e.\, 
there exists  $V\in\cF_o$ s.t. $V^*V=E$ and   $VV^*=1$ ($V$ is an isometry).
\begin{definition}
We take as a 
\textbf{reference representation} for the Dirac field, the representation 
defined by a \textbf{pure quasi-free and gauge invariant Hadamard state}. 
\end{definition} 
As shown above the net $\cF_{KM}$ is 
isotonous, satisfies twisted Haag duality and the local algebras $\cF_o$ are type $\mathrm{III}$ factors.

We define the observable net. To this end 
we consider the \emph{spectral subspaces}
\begin{equation}
\label{eq.Ca5}
\cF_o^n \, := \, 
\{ T \in \cF_o \, : \, U(\zeta) T U(\zeta)^* \, = \, \zeta^n T \, , \, \forall \zeta \in \bU(1) \}^{\prime\prime}
\ \ \ , \ \ \ 
n \in \bZ 
\, . 
\end{equation}
For the grade $n=0$ any $\cF_o^0$ is a von Neumann algebra which is \emph{gauge invariant} and such that  $[ T,S ] = 0$ for $T \in \cF^0_o$ and $S \in \cF^0_a$ with $a \perp o$. This yields the \emph{net of local observables}  $\cF^0_{KM}:KM\ni o\to \cF_o^0\subseteq\cB(\cH)$ which satisfies isotony and causality. The local observable algebras $\cF_o^0$ inherit from $\cF_o$ the properties of being type $\mathrm{III}$ factors, 
see for instance \cite{HdA06}.\smallskip 

%We finally have to take into the game the zero-charge "vacuum" representation of the net of local observables.
%To this end we consider the projection $E_0$ on the vector space
%%

Finally we have to take into the game the zero-charge "vacuum" representation of the net of local observables.
The \emph{vacuum Hilbert space} $\cH_0$ is the subspace
\begin{equation}
\label{eq.Ca6}
\cH_0 \ := \   {\mathrm{closure}} \, \{ A \Omega \, : \, A \in \cF_o^0 \, , \, o \in KM \} \subseteq \cH \, ,
\end{equation}
and the \emph{vacuum representation} $\pi_0:\cF^0_{KM}\to \cB(\cH_0)$ is defined as 
\begin{equation}
\label{vacrep}
\pi_0(A) := AE_0 \ , \qquad   A \in \cF_o^0, \  o \in KM \ ,  
\end{equation} 
defined in terms of the projection $E_0$ onto $\cH_0$. Since $E_0$ is an element of $(\cF^0_o)'$ for any diamond $o$ and since the local algebras $\cF^0_o$ are type $\mathrm{III}$ factors, $\pi_0$ is faithful
and the von Neumann algebras 
\begin{equation}
\label{eq.Ca6a}
\cA_o \, := \, \pi_0(\cF_o^0) \, \subseteq \ \cB(\cH_0) \, , \ \qquad o \in KM \, .
\end{equation}
turn out to be type $\mathrm{III}$ factors. We call the correspondence $\cA_{KM}: KM\ni o\to \cA_o \subseteq \cB(\cH_0)$ \emph{the net of local observables in the vacuum representation}. 
So, $\cA_{KM}$ is an irreducible net satisfying isotony, causality and whose local algebras $\cA_o$ are type $\mathrm{III}$ factors. Furthermore, it is a well known fact \cite{DHR69} that as the net of fields $\cF_{KM}$ satisfies twisted Haag duality, the net $\cA_{KM}$ satisfies  \emph{Haag duality}.

\subsection{Charge transporters of the free Dirac Field}
\label{Cc}
 
Having established that the observable net  $\cA_{KM}$ of the Dirac field in the vacuum representation 
satisfies the same properties as those assumed in Section \ref{Ba}, 
now we construct the charge transporters of the Dirac field, which provide the DHR-sectors for $\cA_{KM}$.\medskip

We start by taking for any $o \in KM$ a normalized spinor
$f_o$ of $\cS_o(DM)$, that is  $(f_o,f_o)_s =1$. With this choice 
any $\psi(f_o)$ turns out to be a partial isometry of $\cF_o$ since, as it is easily verified by \eqref{CARs},  
the operators $\psi(f_o) \psi (f_o)^*$ and  $\psi(f_o)^* \psi (f_o)$ are projections. 
Notice, in particular, that $\psi(f_o)$ is a partial isometry of the spectral subspace $\cF^1_o$, see \eqref{eq.Ca5}.
We can now make $\psi(f_o)$ equivalent to a unitary of the same spectral subspace. To this end 
we note that both the projections $\psi(f_o) \psi (f_o)^*$ and  $\psi(f_o)^* \psi (f_o)$, since gauge invariant \eqref{gauge}, belong to the algebra $\cF^0_o$ 
which is  a type  $\mathrm{III}$ factor. Hence there exist $W_o,V_o\in\cF_o^0$ s.t.    
\[
{V_o}^*V_o=1 \ \ , \ \ V_oV_o^* = \psi(f_o) \psi (f_o)^* \ \ , \ \  
W_o^*W_o=1 \ \ , \ \ W_oW_o^* =  \psi(f_o)^* \psi (f_o) \ . 
\]
%
%This and \eqref{CARs} imply that any $\psi(f_o) \in \cF_o$ is a partial isometry, i.e.\, that 
%$\psi(f_o) \psi (f_o)^*$, $\psi(f_o)^* \psi (f_o)$ are projections. Furthermore 
%as the projections $\psi(f_o) \psi (f_o)^*$, $\psi(f_o)^* \psi (f_o)$ belong to the observable algebra 
%$\cF^0_o$, because  gauge-invariant \eqref{gauge}, and  $\mathrm{III}$ 
% 
%
%
%
%in fact, by \eqref{CARs},
%%
%$0 =  2\psi(f_o)\psi(f_o)$
%%
%and
%%
%\[
%\psi(f_o)^* \psi (f_o) \, \psi(f_o)^*\psi (f_o) = 
%- \psi(f_o)^*\psi(f_o)^*\psi(f_o)\psi (f_o) +\psi(f_o)^*\psi (f_o) =
%\psi(f_o)^*\psi (f_o)  \ , 
%\]
%and similarly 
%%
%$\psi(f_o)\psi (f_o)^*\psi(f_o)\psi (f_o)^* = \psi(f_o)\psi (f_o)^*$. 
%%
%We conclude that the operators 
%%
%$\psi(f_o) \psi (f_o)^*$, $\psi(f_o)^* \psi (f_o)$ 
%%
%are projections and belong to $\cF_o^0$ because are gauge-invariant \eqref{gauge}.
%%
%Since the $\cF_o^0$ are type $\mathrm{III}$ factors, 
%we can make $\psi(f_o) \psi (f_o)^*$, $\psi(f_o)^* \psi (f_o)$ unitaries, by using isometries $W_o,V_o\in\cF_o^0$ such that 
%%
%\[
%V_o^*V_o=1 \ \ , \ \ V_oV_o^* = \psi(f_o) \psi (f_o)^* \ \ , \ \  
%W_o^*W_o=1 \ \ , \ \ W_oW_o^* =  \psi(f_o)^* \psi (f_o) \ ,
%\]
%
The operator 
\begin{equation}
\label{eq.CaLoc}
\varphi_o  :=  W^*_o\psi(f_o)V_o \ , \qquad o\in KM
\end{equation}
is a unitary in $\cF_o$ and an element of $\cF_o^1$ because  by construction
$\eta_\zeta(\varphi_o) = \zeta \varphi_o$, $\zeta \in \bU(1)$. 
The operators $\varphi_o$ do not fulfill the CARs \eqref{CARs}. 
Nevertheless, if $a\perp o$, $A\in \cF^0_o$ and $F\in\cF_a$ then $[A,F]=0$, so that 
\begin{align*}
\varphi_o^*\,\varphi_a & = V^*_o\psi(f_o)^*W_o\, W^*_a\psi(f_a)V_a= V^*_o\psi(f_o)^*\,W^*_a\,W_o \psi(f_a)V_a \\
& = V^*_oW^*_a \psi(f_o)^* \psi(f_a) W_oV_a= - V^*_oW^*_a   \psi(f_a) \psi(f_o)^*  W_oV_a \\
%
%& = - W^*_oW^*_a  \varphi(f_a)V_a \varphi(f_o)^*  V_o \ = - W^*_a  \varphi(f_a)V_a W^*_o\varphi(f_o)^*  V_o \\
%
& = - \varphi_a\,\varphi_o^*
\end{align*}
and similarly for the other CARs. In conclusion, for all $o \perp a$
\begin{equation}
\label{eq.Ca9a}
\varphi_o^* \, \varphi_a \ = \ - \varphi_a \, \varphi_o^* 
\ \ \ , \ \ \ 
\varphi_o \, \varphi_a \ = \ \varphi_a \, \varphi_o
\ \ \ , \ \ \ 
\varphi_o \in (\cF_a^0)^\prime 
\, .
\end{equation}
Now we define the charge transporter $z$ associated with the Dirac field as
\begin{equation}
\label{eq.Ca90}
 z(a,o) \ := \ \pi_0(\varphi_a^*\varphi_o) \ = \ E_0 \varphi_a^*\varphi_o \ ,\qquad o\subseteq a \ ,
\end{equation}
where $\pi_0$ is the vacuum representation.  Clearly $z$ is localized and 
$z(a,o)\in\cA_a$ for any  $o\subseteq a$,  satisfies the 
1-cocycle identity 
\[
z(\tilde o,a)\, z(a,o) \ = \ 
\pi_0(\varphi_{\tilde o}^* \varphi_a \varphi_a^*\varphi_o) \ = \
%
%\varphi_{\tilde o}^*\varphi_o = 
%
z(\tilde o,a) \ , \qquad a \subseteq o \subseteq \tilde o \ , 
\]
and 
\begin{equation}
\label{eq.Ca9}
z(p) =  
z(o,o_{n})\cdots z(o_2,o_1)z(o_1,a)  = \pi_0(  
\varphi_o^*\varphi_{o_{n}} \cdots \varphi_{o_1}^*\varphi_a) \, = \, 
\pi_0(\varphi_o^* \varphi_a)
\end{equation}
for any path  $p : a \to o$, and for $o=a$ we find $z(p) = 1$.
Thus $z$ is topologically trivial, $z \in Z^1_t(\cA_{KM})$.
\begin{lemma}
\label{fermistat}
The following assertions hold true:
\begin{itemize}
\item[(i)] the definition of the 1-cocyle $z$ is independent from the choice of the normalized spinors;
\item[(ii)] $z \in Z^1_\rDHR(\cA_{KM})$ and obeys the Fermi statistics i.e. it has statistical dimension $d(z)=1$ and  statistical phase $\kappa(z) = -1$.
\end{itemize}
\end{lemma} 
\begin{proof}
$(i)$ Let $g_o\in\cS_o(DM)$, $o\in KM$, be another choice of normalized spinors. As above, denote with $S_o$ and $T_o$ the 
isometries of $\cF^0_o$ making the operator $\chi_o:= T^*_o\psi(g_o)S_o$  a unitary of $\cF_o$ and of $\cF^1_o$. Then $z^\prime(a,o):= \pi_0(\chi_a^*\chi_o)$ is a 1-cocycle  of $Z^1_t(\cA_{KM})$ which is unitary equivalent to $z$.
In fact  $t_o:= \pi_0(\chi_o\phi^*_o)$ is a unitary of $\cA_o$ for any $o$ satisfying 
$t_a z(a,o)= \pi_0(\chi^*_a\phi_a)\pi_0(\phi^*_a\phi_o) = \pi_0(\chi^*_a\phi_o)= z^\prime(a,o) t_o$ i.e. 
$t: KM\in o\to t_o\in\cA_o$ is a unitary intertwiner of $(z,z^\prime)$. 

$(ii)$ It is enough to prove that the symmetry intertwiner
$\eps \in (z \times z , z \times z)$ is such that $\eps_a = -1$ for all $a$, where
\[
\eps_a \, := \, z(q)^*\times z(p)^* \cdot \, z(p)\times z(q) \, ,
\]
$p:a\to \tilde o$ and $q:a\to o$ are paths with $o\perp \tilde o$,
and $\times$ is the tensor product of $Z^1(\cA_{KM})$  \cite[Theorem 4.9]{Ruz05}. 
The latter is defined by the expressions
\[
z(q)^* \times z(p)^* := z(q)^* \, z(p_1) z(p)^* z(p_1)^*
\ \ \ , \ \ \ 
z(p)\times z(q) := z(p)\, z(p_1) z(q) z(p_1)^*
\, ,
\] 
with $p_1:o_1\to a$, $o_1\perp |q|$.
According to these relations we may take $o\perp a$, $\tilde o=a$ and $p : a \to a$ the trivial path.
So, $z(p)=1$. Moreover, by (\ref{eq.Ca9}) we have $z(q)= \pi_0(\varphi_o^*\varphi_a)$. 
Finally, we may take $o_1$ causally disjoint from $o$ and $a$ 
and $z(p_1)= \pi_0(\varphi_a^*\, \varphi_{o_1})$. In this way, 
%
% z(p)=1     z(q)= \varphi_o^*\varphi_a     z(p_1)= \varphi_a^*\, \varphi_{o_1}
\begin{align*}
z(p)^* \times z(q)^* \cdot \, z(p)\times z(q) & = 
\pi_0 ( \, \varphi_a^*\, \varphi_o \, \varphi_a^*\, \varphi_{o_1} \, \varphi_{o_1}^*\, \varphi_a  \cdot \,
\varphi_a^*\,\varphi_{o_1} \, \varphi_o^*\varphi_a \, \varphi_{o_1}^*\,\varphi_a \, )\\
& = \pi_0 ( \, \varphi_a^*\, \varphi_o \, \varphi_a^*\,\varphi_{o_1} \, \varphi_o^*\varphi_a \, \varphi_{o_1}^*\,\varphi_a \, ) \\
& = - \pi_0 ( \, \varphi_a^*\, \varphi_o \, \varphi_a^*\,  \varphi_o^*\, \varphi_{o_1} \,\varphi_a \, \varphi_{o_1}^*\, \varphi_a \, )  \\
& =  \pi_0 ( \, \varphi_a^*\, \varphi_o \, \varphi_a^*\,  \varphi_o^*\,\varphi_a\, \varphi_{o_1} \, \varphi_{o_1}^*\, \varphi_a \, ) \\
%\, = \, \pi_0 ( \, \varphi_a^* \, \varphi_o \, \varphi_a^*\,  \varphi_o^*\,\varphi_a\, \varphi_a \, ) \\
& =  - \pi_0 ( \, \varphi_a^*\, \varphi_o \, \varphi_o^*\,\varphi_a^*\, \varphi_a\, \varphi_a \, )
\,  = -1 \, ,
\end{align*} 
having used repeatedly (\ref{eq.Ca9a}) for the relations $o_1 \perp o,a$ and $o \perp a$.
We conclude that $\eps_a = -1$ and this completes the proof.
\end{proof}
Thus the superselection sector associated with $z$ has the same charge quantum number as an electron.  We now 
define the charge transporter associated with the conjugated Dirac field. In analogy to what has been done
for $\psi$ we choose a co-spinor $h_o \in\cS(D^*M)$  with $(h_o,h_o)_c=1$ for any $o\in KM$ and set  
$\bar\varphi_o:= {\bar W}^*_o\bar\psi(h_o) \bar{V}_o$, 
where ${\bar W}^*_o,\bar{V}_o$ are the isometries of $\cF^0_o$ making  $\bar\varphi_o$ a unitary of 
$\cF_o$ and an element of the spectral subspace $\cF^{-1}_o$. Then we define 
\begin{equation}
\label{concocy}
\bar{z}(a,o):= \pi_0(\bar{\varphi}^*_a\,  \bar{ \varphi}_o) \ , \qquad o\subseteq a  \ .  
\end{equation}
In complete analogy to what has been shown for $z$,  one can prove that  $\bar{z}$ is independent from the choice of the normalized 
co-spinor and it is a topologically trivial 1-cocycle obeying the Fermi statistics: $d(\bar z)=1$ and 
$\kappa(\bar z)=-1$. 

We now prove that $\bar z$ is the conjugate of $z$ in the sense of the theory of superselection sectors.
\begin{lemma}
\label{concat}
$\bar{z}$ is the conjugate of $z$ in the category $Z^1_{\rDHR}(\cA_{KM})$.
\end{lemma}
\begin{proof}
This amounts to verifying triviality of the tensor product  $(z\times \bar{z})$ defined as  
\[
(z\times \bar{z}) (o,a) \, := \, z(o,a)\, z(p_1)\bar{z}(o,a)z(p_1)^* \ , \qquad  a \subseteq o \, ,
\]
where $p_1:o_1\to a$, with  $o_1 \perp o,a$, is an arbitrary path as in the previous proof \cite{Ruz05,BR08}. As $z$ and $\bar z$ do not depend on the choice of the normalized spinors and co-spinors, respectively, 
it is enough to prove the above triviality for a suitable choice of the normalized co-spinors. 
So if $f_o$, $W_o$ and $V_o$ are the normalized spinor and isometries defining $\varphi_o$, we  
take $h_o:= f^\dagger_o$ and notice
that $(h_o,h_o)_c= (f^\dagger_o,f^\dagger_o)_c= (f_o,f_o)_s=1$. Furthermore, 
as $\bar\psi(h_o)=\psi(h^\dagger_o)^*=\psi(f_o)^*$,   we may take, for the definition of $\bar\phi$, 
 $\bar W_o:= V_o$ and $\bar V_o:= W_o$. Then   
$\bar\varphi_o:= V^*_o \psi(f_o)^* W_o =  \varphi_o^*$ and  
\begin{align*}
(z\times \bar{z}) (o,a) &  = 
\pi_0 ( \,  \varphi_o^*\, \varphi_a \, \varphi_a^*\, \varphi_{o_1}  \bar\varphi_o^*\,\bar\varphi_a\, \varphi_{o_1}^*\,\varphi_a \, )  = 
\pi_0 (\varphi_o^*\, \varphi_{o_1}  \varphi_o\,\varphi^*_a\, \varphi_{o_1}^*\,\varphi_a \, ) \\
& = \pi_0 ( \, \varphi_o^*\, \varphi_{o_1}\, \varphi_{o_1}^*\,  \varphi_o\,\varphi^*_a\, \varphi_a \, )   = 1 
\end{align*}
%\[
%(z\times \bar{z}) (o,a) \, = \,
%\pi_0 ( \,  \varphi_o^*\, \varphi_a \, \varphi_a^*\, \varphi_{o_1}  \bar\varphi_o^*\,\bar\varphi_a\, \varphi_{o_1}^*\,\varphi_a \, ) \, = \,
%\pi_0 ( \, \varphi_o^*\, \varphi_{o_1}\, \varphi_{o_1}^*\,  \bar\varphi_o^*\,\bar\varphi_a\, \varphi_a \, )  \, = \, 
%\pi_0 ( \, \varphi_o^*\, \bar\varphi_o^*\,\bar\varphi_a\,\varphi_a \, ) \, ,
%\]
where we have used the fact that  $o_1$ is causally disjoint from $a$ and $o$ and that $\varphi_o\,\varphi^*_a$ is an observable. 
\end{proof}

%[ev]
%
\begin{remark}
\label{rem.EC}
Performing tensor powers $z^{\times n}$ and $\bar{z}^{\times n}$, $n \in \bN$,
we obtain a structure $Z_{ec}^1(\cA_{KM})$ of DHR-sectors labeled by $\bZ$, where 
positive integers are associated to $z^{\times n}$ and
negative integers to $\bar{z}^{\times n}$
(we define $z^0$ as the reference representation).  
Using the techniques explained in the present section it is easily seen that $z^{\times n}$ may be defined
starting from products of field operators $\psi(f_1) \cdots \psi(f_n)$, 
and similarly $\bar{z}^{\times n}$ with products of conjugate field operators.
Thus it is natural to interpret $Z_{ec}^1(\cA_{KM})$ as the charge superselection rule.
\end{remark}

\section{Dirac fields in  background flat potentials}
\label{sec.DA}

The Aharonov-Bohm effect concerns the behaviour of electrically charged quantum particles
in a space $M$ where a classical background potential $A$ is defined,
with the property of having vanishing electromagnetic tensor $F := dA$.
This amounts to saying that $A$ is a closed de Rham $1$-form, $dA = 0$,
and we say that $A$ is a flat background potential.
Existence of non-trivial flat background potentials is equivalent to the property of having a
non-trivial de Rham cohomology, $H_{dR}^1(M) \neq 0$. 
By standard results, this implies that $\pi_1(M) \neq 0$.

Passing to a relativistic setting, the scenario we adopt for describing quantum interactions
with the flat background potential $A$ is given a spinor field on a globally hyperbolic spacetime $M$,
fulfilling the Dirac equation with interaction $A$.
In the present section we construct such interacting Dirac fields,
and show that they can be interpreted as superselection sectors of a given reference free Dirac field.

\subsection{Preliminaries on flat potentials}
\label{sec.DA1}

We start by collecting some properties of flat potentials,  closed  de Rham forms $A \in Z_{dR}^1(M)$,
in particular their relations with flat line bundles, and we establish a mapping
$A \mapsto \hat{A}$ with $\hat{A}$ a real 1-cocycle in the cohomology of $KM$.
This will be useful in the discussion of AB-sectors of the observable net of the free Dirac field.

\medskip

As a first step we note that, since diamonds $o \in KM$ are simply connected,
there are $C^\infty$ local primitives
$\phi : o \to \bR$
such that 
\[
d\phi_o(x) = A(x) \ , \qquad x\in o \in KM \, .
\] 
If $a,o \in KM$ and $a \cap o \neq \emptyset$, then 
$\{ d\phi_o - d\phi_a \}(x) = A(x) - A(x) = 0$,
thus the function $\phi_o - \phi_a$ is constant on each connected component of $o \cap a$ 
(that is, $\phi_o - \phi_a : o \cap a \to \bR$ is \emph{locally constant}).
On these grounds, defining 
\begin{equation}
\label{eq.DA1:1}
g_{oa}(x):= e^{i (\phi_o-\phi_a)(x)} \, \in \bU(1) \ , \qquad x\in a\cap o \, ,
\end{equation}
we get a family of locally constant functions fulfilling the cocycle relations 
$g_{oa} g_{ac} = g_{oc}$, as it can be trivially verified
on each overlap $o \cap a \cap c \neq \emptyset$.
The next result is standard in geometry: for a proof we refer the reader to \cite{VasBPS}:
\begin{lemma}
\label{lem.DA0}
Let $A \in Z_{dR}^1(M)$. Then $g := \{ g_{oa} \}$, defined as in (\ref{eq.DA1:1}),
is a set of transition maps for the flat line bundle $\cL_A  :=  \hat{M} \times_\si \bC$
defined as in (\ref{eq.Es}), where 
\[
\si([\ell]) \, := \, \exp i\oint_\ell A \ \ \ , \ \ \ [\ell] \in \pi_1(M) \, .
\]
\end{lemma}
\noindent Let now
\[
\pi_o : \cL_A \restriction_o \, \to M \times \bC \ \ \ , \ \ \ o \in KM \, ,
\]
denote local charts such that $g_{oa}\pi_a = \pi_o$ on $o \cap a$.
The previous Lemma implies that the relation $e^{-i\phi_o} \pi_o = e^{-i\phi_a} \pi_a$ holds for any $a \cap o \neq \emptyset$.
Thus the maps
$e^{-i\phi_o} \pi_o  : \cL_A \restriction_o \to M \times \bC$, $o \in KM$,
glue in the correct way and they induce the bundle isomorphism $\vartheta : \cL_A \to M \times \bC $ defined as 
\begin{equation}
\label{eq.DA1.2}
\vartheta \restriction_o \, := \,  e^{-i\phi_o} \pi_o \, , \qquad o\in KM \, .
\end{equation}
Note that $\vartheta$ locally looks like the multiplication by $e^{-i\phi_o}$.
Let now $a \subseteq o$; then $\phi_o - \phi_a$ is a constant real function on $a$
and we define 
\begin{equation}
\label{eq.DA1}
\hat{A}_{oa} \, := \, \phi_o - \phi_a  \, \in \bR
\ \ , \ \
\hat{A}_{ao} \, := \, - \hat{A}_{oa} 
\ \ , \ \ 
a \subseteq o
\, .
\end{equation}
The relations
$\hat{A}_{co} + \hat{A}_{oa} = \hat{A}_{ca}$
are clearly fulfilled for all $a \subset o \subset c \in KM$,
that is, $\hat{A}$ is a real cocycle in the cohomology of $KM$.
Note that by definition (\ref{eq.DA1:1}), we have
\begin{equation}
\label{eq.DA1:1'}
g_{oa}  \ = \  e^{i\hat{A}_{oa}} \ \ \ , \ \ \ a \subseteq o \, .
\end{equation}
The cocycle $\hat{A}$ can be used to express the holonomy of $A$.
In view of \cite[Eq.4.7]{VasQFT}, given a loop $\gamma : [0,1] \to M$ and $p_\gamma=(o_n,o_{n-1})*\cdots * (o_1,o_0)$ 
a poset approximation of $\gamma$, we have
\begin{equation}
\label{lem.DA1}
\oint_\gamma A  \ =  \ \hat A_{p_\gamma} \ := \ \sum^{n}_{i=1} \hat{A}_{o_{i}o_{i-1}} \, .  
\end{equation}

\subsection{The Dirac field interacting with a background flat potential}
\label{sec.DA2}

We construct a quantum Dirac field interacting with a background flat  electromagnetic potential,
and we describe it in terms of a family of Dirac fields locally defined on diamonds.
These fields yield a ''local coordinate" description of the interacting field, 
and will play an important role at the level of the associated nets of von Neumann algebras.

The interacting field shares with twisted fields in the sense of Isham \cite{Ish78,AI79}
the property of being defined on the space of sections of a twisted bundle. Yet it is not exactly a field of this type as we shall see in the following lines.

\medskip

Let $A = (A^\mu) \in Z_{dR}^1(M)$. Our task is the construction of a Dirac field $\psi_A$ such that
\[
 \psi_A\circ (i\slashed\nabla + \slashed A - m)   =  0 \ , 
\]
where  $\slashed A$ is defined in local frames as $\gamma^\mu_o A_\mu$ 
(note that here the components of $A$ have lower indices),
and since the local Dirac matrices $\gamma^\mu_o$
can be arranged to form a section of $T^*M \otimes {\bf end}DM$, 
%
% Questo fatto viene menzionato nel lavoro \cite{DHP09}. EV
%
we have $\slashed A \in \cS({\bf end}DM)$.

\medskip

\noindent We note, that if $\psi_A$ is given, then for any $o \in KM$ and $f \in \cS_o(DM)$ it turns out
\[
\begin{array}{lcl}
0 \ = \
\psi_A(  (i \slashed\nabla + \slashed A - m) (e^{i\phi_o}f) ) & = &
\psi_A( -\gamma^\mu_o (\partial_\mu \phi_o) e^{i\phi_o}f + i e^{i\phi_o} \slashed\nabla f +  ( \slashed A - m ) (e^{i\phi_o}f) \, ) \  \\ & = &
\psi_A( e^{i\phi_o}\, ( -\slashed A f + i\slashed\nabla f +  (\slashed A - m) f)  \, ) \  \\ & = &
\psi_A( e^{i\phi_o} (( i\slashed\nabla - m) f)  \, ) \, .
\end{array}
\]
Thus applying the local gauge transformation
$\psi_A \to \psi_{A,o} := \psi_A \circ e^{i\phi_o}$,
we find that $\psi_{A,o}$, as a field evaluated on the test space $\cS_o(DM)$, solves the free Dirac equation.

The idea of the previous computation is that $A$ appears as an exact 1--form on each (simply connected) diamond $o \in KM$,
thus it can be represented locally as a local gauge transformation making $\psi_A$ a free Dirac field.
Thus, to construct our field $\psi_A$, it comes natural to reverse the above argument and to start with a free Dirac field $\psi : \cS(DM) \to \cB(\cH)$  and then to define
\[
\psi_o : \cS_o(DM) \to \cB(\cH) \ \ \ , \ \ \ 
\psi_o(f) \, := \, \psi(e^{-i\phi_o}f) \ \ \ , \ \ \ 
o \in KM \, .
\]
This implies
\begin{equation}
\label{eq.DA2.0}
\begin{array}{lcl}
0 \ = \
\psi( ( i \slashed\nabla - m ) (e^{-i\phi_o}f) ) & = &
\psi( \gamma^\mu_o (\partial_\mu \phi_o) e^{i\phi_o}f + i e^{i\phi_o} \slashed\nabla f - m (e^{i\phi_o}f) \, ) 
\\ & = &
\psi( e^{i\phi_o} (( \slashed A + i\slashed\nabla - m ) f ) \, )  = 
\psi_o( (i\slashed\nabla + \slashed A - m ) f \, ) \, ,
\end{array}
\end{equation}
for $f \in \cS_o(DM)$. Thus each $\psi_o$ is a local solution of the Dirac equation with interaction $A$.
Yet for $a \subseteq o$ and $f \in \cS_a(DM)$ we find
$\psi_o(f) = \psi(e^{-i\phi_o}f) = \psi( e^{-i\hat{A}_{oa}} e^{-i\phi_a}f)$,
and we conclude
\begin{equation}
\label{eq.DA2.1}
\psi_o(f) \ = \ e^{-i\hat{A}_{oa}} \, \psi_a(f) \ \ \ , \ \ \ 
f \in \cS_a(DM) \ \ \ , \ \ \
a \subseteq o
\, .
\end{equation}
The above relations show that there is a topological obstruction to gluing the local fields $\psi_o$,
encoded by the cocycle $\hat{A}$. 
To get a globally defined and interacting Dirac field, we introduce the twisted Dirac bundle
\[
D_AM := DM \otimes \cL_A \to M \, ,
\]
which, by Lemma \ref{lem.DA0}, is endowed with "local charts"
$\pi_o : (D_AM)|_o \to DM|_o$, $o \in KM$. Here to be concise, we make an abuse of notation by identifying $\pi_o$ and $id_o \otimes \pi_o$,
where $id_o$ is the identity of $DM|_o$. Note that ${\bf end}(D_AM) \simeq {\bf end}DM$.

In the following lines we collect some simple properties of $D_AM$.
First, we note that sections $\varsigma \in \cS(D_AM)$ are in one-to-one correspondence with families 
$\{ \varsigma_o \}$, where each $\varsigma_o := \pi_o \circ \varsigma|_o$ is a section of $DM|_o$ and the relations
$\varsigma_o = g_{oa} \varsigma_a$
hold for $o \cap a \neq \emptyset$ with $g_{oa}$ defined in (\ref{eq.DA1:1}).
Secondly, we note that tensoring (\ref{eq.DA1.2}) by the identity of $DM$ induces the isomorphism
\begin{equation}
\label{eq.lem.DA2.1}
\vartheta : \cS(D_AM) \to \cS(DM) 
\end{equation}
such that $\vartheta(\varsigma) = e^{-i\phi_o}\varsigma_o$
for all $o \in KM$ and $\varsigma \in \cS(D_AM)$ with $\supp(\varsigma) \subseteq o$.
%
%
%\begin{proof}
%Let $\varsigma \in \cS(D_AM)$. For all $o \in KM$ we define
%%
%$\varsigma_o := \pi_o \circ \varsigma |_o$
%%
%By Lemma \ref{lem.DA0} on $a \cap o$ it turns out
%%
%$\pi_o^{-1} \circ \varsigma_o = \pi_a^{-1} \circ \varsigma_a$
%%
%from which follows that $\varsigma_o = g_{oa} \varsigma_a$ as desired.
%%
%On the converse, if $\{ \varsigma_o \}$ is a family fulfilling the hypothesis of the Lemma,
%then defining $\varsigma|_o := \pi_o^{-1} \circ \varsigma_o$ we obtain the desired section
%of $D_AM$.
%%
%%
%About (\ref{eq.lem.DA2.1}), we consider the bundle isomorphism $D_AM \to DM$ obtained by tensoring
%by the identity (\ref{eq.DA1.2}). With an abuse of notation we denote it by $\vartheta$,
%so that the property claimed in (\ref{eq.lem.DA2.1}) holds by definition of $\varsigma$.
%Note that here, for $\supp(\varsigma) \subseteq o$, we write $\varsigma_o$ for the section of $DM$
%that equals $\pi_o \circ \varsigma |_o$ in $o$ and is zero elsewhere.
%\end{proof}
%
%
In particular, by definition of $\hat{A}$, we have
\begin{equation}
\label{eq.DA2.1a'}
e^{i\hat{A}_{oa}} \cdot id_a \ = \ \pi_o \circ \pi_a^{-1} \ = \ g_{oa} \ \ \ , \ \ \ a \subseteq o \, .
\end{equation}
Thus, if $\varsigma \in \cS_a(D_AM)$, then 
\begin{equation}
\label{eq.DA2.1a}
\varsigma_o \ = \ e^{i\hat{A}_{oa}} \varsigma_a \ \ \ , \ \ \ a\subseteq o  \, .
\end{equation}
Defining
\begin{equation}
\label{eq.DA2.1b}
\psi_{A,a}(\varsigma) \, := \, \psi_a(\varsigma_a) \, = \, \psi(e^{-i\phi_a}\varsigma_a)
\ \ \ , \ \ \ 
\varsigma \in \cS_a(D_AM)
\, ,
\end{equation}
we have that if $a \subseteq o$ then, combining (\ref{eq.DA2.1}) and (\ref{eq.DA2.1a}),
\[
\psi_{A,o}(\varsigma) \, = \, \psi_o(\varsigma_o) \, = \, 
e^{i\hat{A}_{oa}} \psi_o(\varsigma_a ) \, = \,
e^{i\hat{A}_{oa}} e^{-i\hat{A}_{oa}} \psi_a(\varsigma_a) \, = \, 
\psi_{A,a}(\varsigma) \, , \qquad \varsigma \in \cS_a(D_AM)\ .
\]
The previous computation shows that the operator $\psi_{A,a}(\varsigma)$ is independent of the choice of 
$a \in KM$ such that $\supp(\varsigma) \subseteq a$. 
By (\ref{eq.lem.DA2.1}) we find 
$e^{-i\phi_a} \varsigma_a =$ $\vartheta(\varsigma)$,
implying that each $\psi_{A,a}$ extends to the well-defined field
\[
\psi_A : \cS(D_AM) \to \cB(\cH) 
\ \ \ , \ \ \ 
\psi_A := \psi \circ \vartheta
\, .
\]

\begin{remark}
In the previous expression $\psi_A$ appears as a field defined on $D_AM$ rather than $DM$.
This kind of field was introduced and studied by Isham \cite{Ish78,AI79}.
Note that $\psi_A$ fulfills the CARs 
\[
\{ \psi_A(\varsigma_1)^*  ,  \psi_A(\varsigma_2) \}
\ = \
i \int_M \vartheta(\varsigma_1)^\dagger \, S(\vartheta(\varsigma_2)) 
\ \ \ , \ \ \ 
\varsigma_1 , \varsigma_2 \in \cS(D_AM) \, .
\]
\end{remark}

To evaluate the Dirac equation on $\psi_A$
it is convenient to give a description of $\slashed \nabla$ as an operator on sections of $\cS(D_AM)$
explicitly presented as tensors
$\varsigma =$ $f \otimes s$, $f \in \cS(DM)$, $s \in \cS(\cL_A)$.
To this end we define $\boldsymbol{d} := \vartheta^{-1} d \vartheta$, 
where $\vartheta$ is defined in (\ref{eq.DA1.2}) 
and $d : C^\infty(M,\bC) \to \cS(T^*M \otimes \bC)$ is the complex exterior derivative.
In this way we obtain a connection $\boldsymbol{d}$ on $\cL_A$
compatible with the inner product of $\cS(\cL_A)$ in the sense that 
$d( \Lb s,s' \Rb ) =$ $\Lb \boldsymbol{d} s , s' \Rb + \Lb s , \boldsymbol{d} s' \Rb$, $s,s' \in \cS(\cL_A)$.
The Leibniz rule for the Dirac connection leads to the following extension:
\begin{equation}
\label{eq.Lei}
\nabla( f \otimes s ) \, := \, \nabla f \otimes s + f \otimes \boldsymbol{d} s \, \in \cS(D_AM \otimes T^*M) \, .
\end{equation}
Before constructing the associated Dirac operator
we introduce some useful objects.
We start with the normalized section $e^{i\phi} \in \cS(\cL_A)$ obtained by applying the inverse of (\ref{eq.DA1.2}) 
to the constant section $1$ of $M \times \bC$. Note that by definition $e^{i\phi_o} = \pi_o \circ e^{i\phi}|_o$ for all $o \in KM$ (and this justifies our notation). Hence we can write
\begin{equation}
\label{eq.vT}
\vartheta(s) \ = \ \Lb e^{i\phi} , s \Rb  \ \ \ , \ \ \ s \in \cS(\cL_A) \, ,
\end{equation}
having regarded the complex function $\Lb e^{i\phi} , s \Rb$ as a section of $M \times \bC$.
We can write explicitly the complex 1--form 
$X_s := d(\Lb e^{i\phi},s \Rb) \in \cS(T^*M \otimes \bC)$, 
\begin{equation}
\label{eq.Xs}
X_s \ = \ -iA \cdot \Lb e^{i\phi},s \Rb  +  \Lb e^{i\phi}, \boldsymbol{d} s \Rb \, ,
\end{equation}
having used $i \partial^\mu \phi_o = i A^\mu$ and the fact that $\Lb \, \cdot \, , \, \cdot \, \Rb$ is conjugate linear in the first variable.
Also note that $\boldsymbol{d} s \in \cS(\cL_A \otimes T^*M)$, so the pairing $\Lb e^{i\phi}, \boldsymbol{d} s \Rb$ yields a complex 1--form.
Applying a contraction with local Dirac matrices $\gamma^\mu_o$, we get an operator $\slashed X_s \in \cS({\bf end}DM)$,
\[
\slashed X_s(f) \ = \ -i\slashed A f \cdot \Lb e^{i\phi},s \Rb + \Lb e^{i\phi}, \, \slashed f \otimes \boldsymbol{d} s \Rb
\ \ \ , \ \ \ 
f \in \cS(DM)
\, ,
\]
where, in local frames over $o \in KM$, 
\begin{equation}
\label{eq.fos}
\slashed f \otimes \boldsymbol{d} s \ := \ \gamma^\mu_o f \otimes \partial_\mu s \, \in \cS(D_AM) \, .
\end{equation}
Here $\partial_\mu$ has lower indices, thus $\partial_\mu s$ lives in $\cL_A \otimes TM$ 
and contraction by $\gamma^\mu_o$ drops components in $T^*M$ and $TM$.
Moreover, pairing by $e^{i\phi}$ yields $\Lb e^{i\phi} , \, \slashed f \otimes \slashed ds \Rb \in \cS(DM)$.
With the notation (\ref{eq.fos}), pairing (\ref{eq.Lei}) with the Dirac matrices yields
\begin{equation}
\label{eq.Lei'}
\slashed \nabla( f \otimes s ) \ = \ 
\slashed \nabla f \otimes s + \slashed f \otimes \boldsymbol{d} s  \, .
\end{equation}
In conclusion,
\begin{align*}
\slashed\nabla \circ \vartheta(f \otimes s) & = 
\slashed\nabla ( f \cdot \Lb e^{i\phi},s \Rb )  =
(\slashed\nabla f) \cdot \Lb e^{i\phi},s \Rb  ) + \slashed X_s(f) \\
& =
\slashed\nabla f \cdot \Lb e^{i\phi},s \Rb - i\slashed A f \cdot \Lb e^{i\phi},s \Rb + \Lb e^{i\phi}, \, \slashed f \otimes  \boldsymbol{d}s  \Rb  = 
\vartheta ( \,  (\slashed\nabla - i \slashed A ) f  \otimes s +  \slashed f \otimes \boldsymbol{d}s  \, ) \\
&  = 
\vartheta ( \,  (\slashed\nabla - i \slashed A)  (f \otimes s) \, ) \, ,
\end{align*}
having written $\slashed A (f \otimes s) := ( \slashed A f) \otimes s$ and used 
(\ref{eq.vT}),(\ref{eq.Xs}), (\ref{eq.Lei'}).
This implies
\begin{equation}
\label{eq.DA2.4}
(i\slashed\nabla + \slashed A - m) \psi_A(f \otimes s)   = 
\psi( \vartheta ( \, (i\slashed\nabla + \slashed A - m) \, (f \otimes s)  \, ))  =
\psi((i\slashed\nabla-m)(\vartheta(f \otimes s)))  =
0  ,
\end{equation}
as expected. 
Note that if one starts with the field $\psi_A$, then defining for all $o \in KM$
\[
\psi_o : \cS_o(DM) \to \cB(\cH) 
\ \ \ , \ \ \ 
\psi_o(f) \, := \, \psi_A(\pi_o^{-1}(f))
\]
we get a family of fields fulfilling (\ref{eq.DA2.1}) (see (\ref{eq.DA2.1a'})). Moreover, the chain of identities (\ref{eq.DA2.0}) shows that each field $\psi_o \circ e^{i\phi_o}$
fulfills the free Dirac equation and, again by (\ref{eq.DA2.1a'}), the family $\{ \psi_o \circ e^{i\phi_o} \}$
glue in the correct way to define a free Dirac field $\psi := \psi_A \circ \vartheta^{-1}$.
In conclusion:
\begin{itemize}
\item We accomplished our task of constructing a Dirac field $\psi_A$ interacting with $A$,
      by paying the price of switching from $DM$ to the twisted Dirac bundle $D_AM$;
\item $\psi_A$ is equivalently described by the family $\{ \psi_o \}$ fulfilling (\ref{eq.DA2.1})
      and such that each $\psi_o$ is gauge-equivalent to a free field $\psi$.
\end{itemize}

\subsection{Charge transporters and Aharonov-Bohm effect}
\label{sec.DA3}

In the present section we compare the net defined by $\psi$ and the ones 
defined by the interacting fields $\psi_A$, 
verifying that they all define the same observable net $\cA_{KM}$. 
We show that any interacting field defines an irreducible 1-cocycle
in $Z^1_{\rAB}(\cA_{KM})$ obeying the Fermi statistics and having topological dimension one.

\smallskip

As a first step, recalling (\ref{eq.DA2.1b}), we define
\[
\cF_{A,o} \, := \, \{ \psi_o(f) \, , \, f \in \cS_o(DM) \}^{\prime\prime} \ \ \ , \ \ \ o \in KM \, .
\]
Clearly $\cF_{A,o} = \cF_o$, where $\cF_o$ is the von Neumann algebra generated by the free field $\psi$.
The additional information carried by $A$ is, instead, encoded by (\ref{eq.DA2.1}), 
which imposes to define the *-monomorphisms $\jmath_{oa} : \cF_{A,a} \to \cF_{A,o}$,
\begin{equation}
\label{eq.FA}
\jmath_{oa}(T) \ := \ 
\eta(e^{-i\hat{A}_{oa}}) \, (T) \ = \  
U(e^{-i\hat{A}_{oa}}) \, T \, U(e^{i\hat{A}_{oa}})\ , \qquad  
a \subseteq o
\, ,
\end{equation}
where $\eta$ is the gauge action (\ref{gauge0}) and $U$ is the corresponding unitary action on the Fock space \eqref{gauge}. 
This definition allows to recover the relation between the locally defined fields $\psi_o$,
\[
\jmath_{oa}(\psi_a(f)) \ = \ e^{-i\hat{A}_{oa}} \psi_a(f) \ = \ \psi_o(f) \ \ \ , \ \ \ f \in \cS_a(DM) \, ,
\]
and by applying (\ref{eq.DA1}) and subsequent remarks, we find
\begin{equation}
\label{eq.Da3.2}
\jmath_{oa} \circ \jmath_{ae} \ = \ \jmath_{oe} \ \ \ , \ \ \ e \subseteq a \subseteq o \, .
\end{equation}
The pair $(\cF_A,\jmath)$ is called the \emph{field net twisted by $A$} and, at the mathematical level, it
defines a \emph{precosheaf} \cite{VasQFT}. These objects are more general than nets since,
instead of the inclusion maps, we have the non-trivial *-monomorphisms $\jmath_{oa}$.
As a consequence, 
if $T_o \in \cF_{A,o}$ and $T_a \in \cF_{A,a}$, then the correct way to perform their product is
\begin{equation}
\label{eq.Da3.1}
T_o \jmath_{oa}(T_a) \, \in \cF_{A,o} \, .
\end{equation}
We define the gauge-invariant von Neumann algebras
\[
\cF_{A,o}^0 \, := \, \{  S \in \cF_{A,o} \, : \, U(\zeta) S = S U(\zeta) \, , \, \forall \zeta \in \bU(1) \}
\ \ , \ \
o \in KM \, .
\]
This implies $\jmath_{oa}(S) = S$ for all $S \in \cF_{A,a}^0$, thus $\cF_A^0$ is a \emph{net},
coinciding with the gauge-invariant subnet $\cF^0$ of the free Dirac field since $\cF_{A,o}^0 = \cF_o^0$ for all $o \in KM$.

We now pass to the construction of charge transporters starting from the twisted field net $(\cF_A,\jmath)$.
To this end, we note that the operators $\varphi_o$, $o \in KM$,
defined in (\ref{eq.CaLoc}), can be regarded as unitaries
$\varphi_o$ of  $\cF_{A,o}$
carrying charge $1$. To construct the associated charge transporters, we invoke (\ref{eq.Da3.1}) and define
\begin{equation}
\label{eq.Da3.3}
z(a,o)  :=  \pi^0(\varphi_o^* \, \jmath_{oa}(\varphi_a))  = \pi^0( e^{-i\hat{A}_{oa}} \varphi_o^* \varphi_a)
\  , \qquad 
a \subseteq o \, .
\end{equation}
Here $z(a,o)$ is a unitary of $\cA_o$ and by \eqref{eq.Da3.2} it fulfills the cocycle relations,
\[
z(a,o) \, z(o,e) \ = \ z(a,e) \ \ \ , \ \ \ a \subseteq o \subseteq e \, ,
\]
thus we get  a $1$-cocycle   $z \in Z^1(\cA_{KM})$. \smallskip

%\begin{equation}
%\label{eq.Da3.3}
%z^0(o,a) \, := \, \varphi_o^* \jmath_{oa}(\varphi_a) \, = \, e^{-i\hat{A}_{oa}} \varphi_o^* \varphi_a
%\ \ \ , \ \ \ 
%a \subseteq o \, .
%\end{equation}
%%
%%Note that this definition is made in terms of the charged unitaries $\varphi_o$ and the *-monomorphisms $\jmath_{oa}$.
%%
%By (\ref{eq.Da3.2}) we find that $z^0 \in \cU(\cF_o^0)$ fulfils the cocycle relations,
%%
%\[
%z^0(o,a) \, z^0(a,e) \ = \ z^0(o,e) \ \ \ , \ \ \ e \subseteq a \subseteq o \, ,
%\]
%%
%thus $z^0$ is a 1--cocycle of the gauge-invariant net $\cF^0_{KM}$. 
%%
%At this point we consider the projection $E_0 \in \cB(\cH)$ on the $\bU(1)$-invariant space (\ref{eq.Ca6})
%and the observable net $\cA_{KM}$ (\ref{eq.Ca6a}).
%Defining 
%%
%\begin{equation}
%\label{eq.Da3.4}
%z(o,a) \ := \ E_0 \, z^0(o,a) \, \in \cU(\cA_o) \ \ \ , \ \ \ a \subseteq o \, ,
%\end{equation}
%%
%we get the 1--cocycle 
%%
%$z \in Z^1(\cA_{KM})$.
%%
%By (\ref{Bd:1}), for all $a \subseteq \tilde a$ we have
%%
%\begin{equation}
%\label{eq.Da3.4a}
%z_t(\tilde a,a) \ = \ E_0 \, \lambda_{\tilde a}^{-1} \lambda_a  \varphi_{\tilde a}^* \varphi_a
%\ \ \ , \ \ \ 
%\lambda_a \, := \, e^{ -i\sum_i \hat{A}_{o_{i-1}o_i}} \, \in \bU(1)
%\, ,
%\end{equation}
%%
%where $p_{oa} = (o_n=o , o_{n-1}) \ldots (o_{i-1} , o_i) \ldots (o_2 , o_1=a)$
%are the paths in (\ref{Bd:1}).
%%
Recalling the definitions of charged and topological component of a 1-cocycle given in Section \ref{Bb}, we can 
now prove the following theorem which yields a converse for Corollary \ref{Bd:9a2}. 
\begin{theorem}
\label{thm.D1}
Let $\cA_{KM}$ denote the observable net of the free Dirac field and let $A \in Z_{dR}^1(M)$. Then the
cocycle $z$  defined by \eqref{eq.Da3.3}  is an element of the category $Z^1_\rAB(\cA_{KM})$  satisfying  
\begin{equation}
      \label{eq.thm.D1.3}
      z(o,a) \ = \ \exp \Big( - i\oint_{\ell_{(o,a)}} A \Big) \cdot \, z_c(o,a) 
      \ \ \ , \ \ \ 
      o \subseteq a \, .
      \end{equation}
The charged component $z_\mathrm{c}$ obeys the Fermi statistics, 
and the holonomy of $z$ is given by 
\begin{equation}
\label{eq.thm.D1.1}
z(p_\ell) \ = \ \exp \Big( - i\oint_\ell A \Big) \, \bI \ , \qquad  [\ell] \in \pi_1(M)  \, .
\end{equation}
\end{theorem}
\begin{proof}
The holonomy of $z$ is readily computed by using (\ref{lem.DA1}), (\ref{eq.Ca9}) and (\ref{eq.Da3.3}):
\[
z(p_\ell) \ := \ 
z(o_n,o_{n-1}) \cdots z(o_1,o_0) \ = \ 
e^{ - i\hat{A}_{p_\ell}} \, \bI \ = \
\exp \Big\{ - i\oint_\ell A \Big\} \, \bI \, ,
\]
where $p_\ell = (a,o_{n-1})*\cdots * (o_1,a)$ is a path-approximation of the loop $\ell : x \to x$, $x \in a \in KM$. Concerning the first part of the theorem, given a path frame $P_e$  and recalling the definitions of topological 
and charged component given in Section \ref{Bc} we have 
\[
z_{\mathrm{c}}(o,a)= z(p_{oe}*p_{ea})= \pi^0\big(e^{i (\hat A_{p_{oe}}- \hat A_{p_{ae}})} \varphi^*_{o}\varphi_a\big) = e^{i (\hat A_{p_{oe}}- \hat A_{p_{ae}})} 
\pi^0(\varphi^*_{o}\varphi_a)  \ . 
\]
Defining the unitary $t_a:= e^{-i\hat A_{p_{ae}}}\bI$ for any $a\in KM$, we find 
$t_{o}z_{\mathrm{c}}(o,a)= \pi^0(\varphi^*_{o}\varphi_a) t_a$ for any inclusion $a\subseteq o$. So $z_c$ 
is unitary equivalent  to the DHR cocycle \eqref{eq.Ca90} obeying the Fermi statistics. 
Finally, from the splitting formula \eqref{Bd:3}  and from the 
above relation we get 
\begin{align*}
z(\tilde a,a) &= \alpha_{\tilde ae}\big(u_z(\tilde a,a)\big)z_c(\tilde a a)= \alpha_{\tilde ae}\big(z(p_{e\tilde a}*(\tilde a,a)*p_{ae})\big)\,z_c(\tilde a a) \\
&  =  \alpha_{\tilde ae}\Big(\exp\Big(-i\oint_{\ell_{oa}} A\Big) \bI\Big)\, z_c(\tilde a, a) 
 = \exp\Big( -i\oint_{\ell_{oa}} A\Big) \, z_c(\tilde a, a) \, .
\end{align*}
\end{proof} 

We stress that the connection 1--form $A^z$ defined by $z$ in the sense of Corollary \ref{Bd:9a} is gauge equivalent
to $A$ up to a singular cocycle, that is, $A^z$ stands in the same de Rham class of $A$ modulo a cohomology class $\xi^z \in H^1(M)$. In fact, by iterating (\ref{eq.Bd:9a}) over the path approximation $p_\ell$, we obtain
\[
z(p_\ell) \ = \
\exp \Big( -i\sum_i \oint_{\ell_{(o_{i}o_{i-1})}} A^z \Big) \cdot z_\rc(p_\ell) \ = \
\exp -i\oint_\ell A^z \cdot \bI
\, ,
\]
having used topological triviality of $z_\rc$, that is, $z_\rc(p_\ell) = \bI$.
Thus by applying (\ref{eq.thm.D1.1}), we arrive to
\begin{equation}
\label{eq.thm.D1.4}
\exp \Big( - i\oint_\ell A \Big) \ = \ \exp \Big(-i\oint_\ell A^z\Big)
\end{equation}
for all loops $\ell$. 
Since the exponential map has kernel $\bZ$, and since any singular 1--cycle can be interpreted as a loop, we find
$(2\pi i)^{-1} \oint_c A \, {\mathrm{mod}} \bZ  =$ $(2\pi i)^{-1} \oint_c A^z \, {\mathrm{mod}} \bZ$,
for all 1--cycles $c$. The previous equality says that $A$ and $A^z$ define the same differential character 
in the sense of Cheeger and Simons, thus by the third exact sequence of \cite[Theorem 1.1]{CS85} we
conclude that the de Rham cohomology class of $A - A^z$ defines by integration the desired class $\xi^z \in H^1(M)$.

\subsection{Twisted nets and cocycles}
\label{sec.DA3'}

Indeed, the notion of twisted net can be given for arbitrary families $\sigma_{oa} \in \bU(1)$, $a \subseteq o \in KM$, fulfilling the cocycle relations
$\sigma_{oa}  \sigma_{ae} \ = \ \sigma_{oe}$.
By the results in \cite{RRV07}, this is equivalent to giving the morphism 
\[
\sigma : \pi_1(M) \to \bU(1) 
\ \ \ , \ \ \ 
\sigma([\ell]) \, := \, \sigma(p_\ell) \, := \, \sigma_{ao_{n-1}} \cdots \sigma_{o_1a}
\, ,
\]
where $p = (a,o_{n-1}) * \ldots * (o_1,a)$ is a path-approximation of the loop $\ell : [0,1] \to M$.
We call $\sigma$ a \emph{topological twist} of $\cF_{KM}$.
Defining $\jmath_{oa} := \eta(\sigma_{oa}) \restriction \cF_a$ we get the desired twisted field net, that we denote by 
$(\cF_\sigma,\jmath)$, $\cF_{\sigma,o} := \cF_o$, $\forall o \in KM$.
By the results in \cite{RVCX,RVkhom}, a twisted field net is nothing but
a representation of $\cF_{KM}$ on the net of Hilbert spaces $U_{oa} : \cH_a \to \cH_o$, $a \subseteq o$, 
with $\cH_o \equiv \cH$, $o \in KM$, and $U_{oa} := U(\sigma_{oa})$, $o \subseteq a$,
or, equivalently, a representation of $\cF_{KM}$ over the flat Hilbert bundle $H \to M$ with fibre $\cH$
and monodromy $U \circ \sigma : \pi_1(M) \to \cU(\cH)$.
%
%The case $\sigma := \exp i\oint A$ yields the motivation for introducing twisted field nets.

\medskip

Let us now return on the observable net $\cA_{KM}$ and consider the charge 1 DHR-sector (\ref{eq.Ca90}),
that here we denote by
$z_1(o,a) := \pi_0(\varphi_o^* \varphi_a)$, $a \subseteq o$.
\begin{theorem}
\label{thm.D2}
Let $\cA_{KM}$ denote the observable net of the free Dirac field. 
Then there is a one-to-one correspondence between sectors $z \in Z_\rAB^1(\cA_{KM})$ with charge component $z_1$ 
and topological dimension 1, and twisted field nets $(\cF_\sigma,\jmath)$. 
%The correspondence is defined by assigning to each $z$ the holonomy 
%$\sigma : \pi_1(M) \to \bU(1)$ appearing in (\ref{Bd:7}).
\end{theorem}

\begin{proof}
Let $\sigma_{oa} \in \bU(1)$, $a \subseteq o$, denote a cocycle and
$\sigma : \pi_1(M) \to \bU(1)$ denote the associated topological twist of $(\cF_\sigma,\jmath)$.
Then in accord with (\ref{eq.Da3.3}) we define
\[
z(a,o) \, := \, 
\pi^0(\varphi_o^* \jmath_{oa}(\varphi_a)) \, = \, 
\sigma_{oa} z_1(a,o)
\ \ , \ \ 
o \subseteq a
\, .
\]
A straightforward verification then shows that $z \in Z_\rAB^1(\cA_{KM})$ has the desired properties.
Conversely, let $z \in Z_\rAB^1(\cA_{KM})$ as in the statement of the Theorem. 
For all $[\ell] \in \pi_1(M)$ we set
$\sigma([\ell]) := {\mathrm{hol}}_{A^z}(\ell) \in \bU(1)$
and we define the associated twisted field net $(\cF_\sigma,\jmath)$. The cocycle $z'$ defined by $\sigma$ as in the previous part
of the previous theorem is then given by
\[
z'(o,a) \ = \ {\mathrm{Hol}}_{D^z}(\ell_{(o,a)}) \cdot z_1(o,a) \, ,
\]
and Corollary \ref{Bd:9a} implies $z=z'$.
\end{proof}

Considering in particular topological twists arising from closed 1--forms
we obtain the following Theorem, which resumes the results of the last two sections: 
\begin{theorem}
\label{thm.D1'}
Let $\cA_{KM}$ denote the observable net of the free Dirac field.
Then for any closed de Rham form $A \in Z_{dR}^1(M)$ there are:
\begin{enumerate}
\item A twisted Dirac field $\psi_A : \cS(D_AM) \to \cB(\cH)$ solving the Dirac equation with interaction $A$;
\item A twisted field net $(\cF_A,\jmath)$ with topological twisting $\sigma := \exp-i\oint A$;
\item A cocycle $z_A \in Z_\rAB^1(\cA_{KM})$ with charge component $z_1$ and holonomy 
      \begin{equation}
      \label{eq.thm.D1.1'}
      z_A(p_\ell) \ = \ \exp \Big\{ - i\oint_\ell A \Big\} \, \bI \ , \qquad [\ell] \in \pi_1(M)  \, ,
      \end{equation}
      where $p_\ell$ is a path approximation of $\ell$.
\end{enumerate}
If $H_1(M)$ has no torsion elements, then any irreducible $z \in Z_\rAB^1(\cA_{KM})$ 
with charge component $z_1$ and topological dimension 1 is equivalent to a cocycle of the type $z_A$.
\end{theorem}

\begin{proof}
The twisted Dirac field $\psi_A$ has been constructed in \S \ref{sec.DA}.
Starting from $\psi_A$, we constructed in (\ref{eq.FA})
the twisted field net $(\cF_A,\jmath)$, having the required twisting.
The cocycle $z_A$ is obtained by the previous theorem and has the form
\[
z_A(o,a) \ = \ \exp \Big\{ -i \oint_{\ell_{(o,a)}} A \Big\} \cdot \, z_c(o,a) 
\ \ \ , \ \ \ 
o \subseteq a \, ,
\]
from which (\ref{eq.thm.D1.1'}) follows.
Finally, if $H_1(M)$ has no torsion, then by the considerations after (\ref{eq.expA})
we conclude that any holonomy $\sigma : \pi_1(M) \to \bU(1)$ is of the desired type $\exp -i\oint A$.
\end{proof}

\section{Conclusion and outlooks} 
\label{sec.E}

In the present paper we studied the superselection sectors defined on curved spacetimes \cite{BR08}
by analyzing them in the specific model of the observable net of the free Dirac field.
We provided the expected physical interpretation of the sectors,
showing that they correspond to background flat potentials interacting with Dirac fields.
%which locally looks like the free Dirac field.
The correspondence is essentially one-to-one, 
because potentials having the same Aharonov-Bohm phase define the same sector,
and the phase is the actual observable quantity.
The picture is then completed in terms of twisted field nets, 
whose inclusion morphisms are designed to reproduce the relative phases
that appear when one tries to describe the interacting Dirac field in terms of "local charts" (\ref{eq.DA2.1}).
Twisted field nets are, on turns, in one-to-one correspondence with superselection sectors.

\medskip
 
For the physical interpretation of our results 
a crucial point is that the background flat potentials can be reconstructed having merely as input 
the localized loop observables (\ref{Bc:1}) defined by the sectors.
Thus the potential is a byproduct of the sector, 
which, as well-known, corresponds to a state of the observable net,
in general affected by the spacetime topology.
We then conclude that the potential is codified in the preparation of the state.

\medskip

\noindent In particular, in the case of the classical Aharonov-Bohm effect, we argue that:
\begin{itemize}
\item When the experimenter shields the solenoid, the space where the charged particles are confined acquires
      a non-trivial topology, with fundamental group $\bZ$.
\item Switching on the magnetic field $\vec{B}$ inside the solenoid makes the system fall into a superselection sector of AB type,
      labeled by the Aharonov-Bohm phase $\gamma \mapsto \exp -i \oint_\gamma A$.
      This fits the abstract discussion by Morchio and Strocchi \cite{MS07}, as well as old results in path-integral quantization
      (see \cite{Hor80} and references cited therein).
      Thus the presence of the non-trivial potential $A$, $d\vec{A}=\vec{B}$, is regarded as part of the preparation of the state:
      if the experimenter switches off the magnetic field, $\vec{B} = 0$, then we have a topologically trivial sector.
\end{itemize}
The lesson that we learn is that $A$ is interpreted as a generalized DHR-charge of Aharonov-Bohm type,
and can be reconstructed using exclusively localized observables.

\medskip

Finally we remark that for the construction of our interacting  field $\psi_A$ it is essential that the twisting bundle $\cL_A$ is topologically trivial. 
Yet in general the twisting bundle 
\[
\cL_z \, := \, \hat{M} \times_\sigma \bC^n \, ,
\]
where $\sigma : \pi_1(M) \to \bU(n)$ is as usual the holonomy representation defined by $z \in Z^1_\rAB(\cA_{KM})$,
is non trivial. This may occur when:
\begin{enumerate}
\item $n=1$ and the homology $H_1(M)$ has torsion
      {\footnote{For examples, see \cite{Fre86}. EV thanks M. Benini and A. Schenkel for remarking this fact.}};
\item $n>1$, a case that appears for $\pi_1(M)$ non-Abelian or for Dirac fields with non-Abelian gauge group. Under this hypothesis the argument for proving triviality of $\cL_A$ fails.
      We note that this is the case of explicit physical interest: for example, two parallel solenoids in the Aharanov-Bohm apparatus 
      yield $\pi_1(M)$ isomorphic to the free group with two generators.
\end{enumerate}
Even if there exists no problem in defining the twisted field net $(\cF_\sigma,\jmath)$ in these cases,
we believe that it is desirable from the point of view of physical interpretation to construct the interacting field in correspondence of the sector,
thus this point is object of a work in progress.
One possible approach may be to embed $\cL_z$ into some trivial bundle $M \times \bC^m$, which always exists for $m$ great enough. 
As an alternative we may directly take twisted Dirac bundles of the type
$D_zM \, := \, DM \otimes \cL_z$,
and construct a twisted Dirac field 
$\psi_z : \cS(D_zM) \to \cB(\cH)$
in sense of Isham.

%%%%%%%%%%%%%%%%%%%%%%%%%%%%%%%%%%%%%%%%%%%%%%%%%%%%%%%%%%%%%%%%%%%%%%%%

\appendix 

\section{Results on topological sectors}
\label{App}
This appendix is devoted to proving that one can avoid using punctured Haag duality in the analysis of superselection sectors; 
only Haag duality is enough.\smallskip

We consider the observable net $\cA_{KM}$ defined in a representation satisfying Haag duality. The key property 
of 1-cocycles that we need in this appendix is the following: for any pair $a,a_1\in KM$ in the causal complement of $o\in KM$ one has that \cite[Corollary 1]{BR08}  
\begin{equation}
\label{app:1}
z(q)\in\cA^\prime_o \ , \qquad  q:a\to a_1 \ .
\end{equation}
which is a consequence of homotopy invariance of 1-cocycles and of pathwise connectedness of the causal complement of a diamond. 
\begin{lemma}
\label{app:2}
Let $\cA_{KM}$ be a causal net satisfying Haag duality. Given an inclusion of diamonds $o\subseteq \tilde o$, 
let $a\perp \tilde o $ and $p:a\to o$. Then the following properties hold: 
\begin{itemize}
\item[i)]  The adjoint action $z(p) \cA_{\tilde o } z(\bar p)$ 
is independent of the choice of $p$ and $a$.
\item[ii)]  $z(p) \cA_{\tilde o} z(\bar p) \subseteq \cA_{\tilde o}$. 
\end{itemize}
\end{lemma}
\begin{proof}
$(i)$ If $q: a_1\to o$ with $a_1\perp \tilde o $ observe that 
\[
z(p) A z(\bar p) = z(q) z(\overline {q}*p) A  z(\bar p* q) z(\bar q) = z(q) A z(\bar q) \ , \qquad A\in \cA_o
\] 
where we have used \eqref{app:1} as $\bar q*p:a\to a_1$ and $a,a_1$ are in the causal complement
of $\tilde o$. $(ii)$ Here we use Haag duality. Take $o_1\perp \tilde o$ and $B\in\cA_{o_1}$.  
Since $o_1\perp \tilde o$ we can find $a_1$ with $a_1\perp \tilde{o}, o_1$  and a path $q:a_1\to o$. By the just proved result we have that for any $A\in\cA_{\tilde o}$ 
\[
z(p) A z(\bar p) B = z(q) A z(\bar q) B  = z(q) A B z(\bar q) = z(q) BA z(\bar q) = Bz(q) A z(\bar q)= Bz(p) A z(\bar p)
\]
where, again,  we have used \eqref{app:1} as $q:a_1\to o$ and $a_1$, $o$ are in the causal complement of $o_1$.
\end{proof}
We can now prove \cite[Lemma 4.5(i)]{Ruz05}.  
\begin{proposition}
\label{app:3}
Let $\cA_{KM}$ be a causal net satisfying Haag duality. Given $o,a\in KM$ with $o\perp a$. Then, for any 
inclusion $\tilde a\subseteq a$ and any path $p:\tilde a\to o$  there holds
\[
z(p) \cA'_a z(\bar p)\subseteq \cA'_a \ . 
\]
\end{proposition}
\begin{proof}
Take $B\in\cA_a$,  $A\in\cA^\prime_a$ and $p$ as in the statement. By assumption we have that 
$\bar p:o\to \tilde a $ and that $o\perp \tilde a$. By the previous Lemma 
we have that $z(\bar p) B z(p)\in \cA_a$;
hence 
\[
z(p) A z(\bar p) B = z(p) A z(\bar p) B z(p) z(\bar p) = 
 z(p) z(\bar p) B z(p) A z(\bar p) = B z(p) A z(\bar p) \ ,   
\] 
completing the proof. 
\end{proof}
This result implies that \eqref{Bd:3a-1} gives an endomorphism
$\rho^z(o)_{\tilde a}:\cA^\prime_a\to \cA^\prime_a$ for any $\tilde a\subseteq a$ and $o\perp a$
and that we can reply the analysis of superselection sectors without making use of punctured Haag duality. \smallskip

\noindent \textbf{Acknowledgements.} G. Ruzzi acknowledges the MIUR Excellence Department Project awarded
to the Department of Mathematics, University of Rome Tor Vergata, CUP E83C18000100006.

\end{document}